\documentclass[11pt]{article}

\usepackage[a4paper]{geometry}

\usepackage[english]{babel}
\usepackage[utf8]{inputenc}
\usepackage[T1]{fontenc}
\usepackage{epsfig, float}
\usepackage{appendix}

\usepackage{mathcomp}
\usepackage{dsfont}
\usepackage{lmodern}
\usepackage{scrextend}
\usepackage{color}

\usepackage[all]{xy}

\usepackage[nottoc]{tocbibind}

\usepackage{amsmath,amssymb,amsfonts,amsthm}
\usepackage{amsfonts}
\usepackage{url}
\usepackage{bbm}
\usepackage{mathrsfs} 

\usepackage{enumerate}
\usepackage[hidelinks]{hyperref}

\newcommand{\norm}[1]{\left\lVert#1\right\rVert}

\newcommand{\R}{\mathbb{R}}
\newcommand{\C}{\mathbb{C}}
\newcommand{\N}{\mathbb{N}}

\newcommand{\supp}{{\rm supp} \, }

\renewcommand{\Re}{\operatorname{Re}}
\renewcommand{\Im}{\operatorname{Im}}

\newtheorem{theorem}{Theorem}[section]
\newtheorem{lemma}[theorem]{Lemma}
\newtheorem{proposition}[theorem]{Proposition}
\newtheorem{remark}[theorem]{Remark}
\newtheorem{definition}[theorem]{Definition}
\newtheorem{corollary}[theorem]{Corollary}
\newtheorem{assumption}[theorem]{Assumption}

\numberwithin{equation}{section} 

\usepackage{cancel}

\usepackage[normalem]{ulem}
\usepackage{todonotes}

\definecolor{miguelscolor}{rgb}{.7,.2,.2}
\definecolor{dirkscolor}{rgb}{.2,.2,.7}
\definecolor{felixscolor}{rgb}{.2,.7,.2}
\definecolor{felixscolor1}{rgb}{.7,.2,.7}

\newcommand{\ac}[1]{}

\date{\today}

\title{\Large{\textsc{ One-Boson Scattering Processes\\in the massive
Spin-Boson Model}}}

\author{Miguel Ballesteros\thanks{\texttt{miguel.ballesteros@iimas.unam.mx},
        Instituto de Investigaciones en Matem\'aticas Aplicadas y en Sistemas,
        Universidad Nacional Aut\'anoma de M\'exico}, Dirk-Andr\'e
        Deckert\thanks{\texttt{deckert@math.lmu.de}, Mathematisches Institut
        der Ludwig-Maximilians-Universität München}, J\'er\'emy
        Faupin\thanks{\texttt{jeremy.faupin@univ-lorraine.fr},  Institut Elie
        Cartan de Lorraine, Universit\'e de Lorraine},  Felix
    H\"anle\thanks{\texttt{haenle@math.lmu.de}, Mathematisches Institut der
Ludwig-Maximilians-Universität München}}

\begin{document}
\maketitle

\begin{abstract}
    The Spin-Boson model describes a two-level quantum system that interacts
    with a second-quantized  boson  scalar field.  Recently the relation between the
    integral kernel of the scattering matrix and the resonance in this model
    has been established in~\cite{bdh-scat} for the case of massless bosons.
    In the present work, we treat the massive case. 
    On the one hand, one might rightfully expect that the
    massive case is easier to handle since, in contrast to the massless case,
    the corresponding Hamiltonian features a spectral gap.  On the other hand,
    it turns out that the non-zero boson mass introduces a new complication as
    the spectrum of the complex dilated, free Hamiltonian exhibits  lines of 
    spectrum attached to every multiple of the boson rest mass energy starting
    from the ground and excited state energies. This leads to an absence of
    decay of the corresponding complex dilated resolvent close to the real
    line, which, in~\cite{bdh-scat}, was a crucial ingredient to control the
    time evolution in the scattering regime. With the new strategy presented
    here, we provide a proof of an analogous formula for the scattering
    kernel as compared to the massless case and use the opportunity to provide
    the required spectral information by a Mourre theory argument  combined with a suitable application of the Feshbach-Schur map  instead of
    complex dilation. 
\end{abstract}

\section{Introduction}
\label{sec:introduction}

The Spin-Boson model is a widely employed model in quantum field theory  that describes the interaction between a two-level quantum
system and a second-quantized scalar field. The model is interesting as it
shares many important features of, e.g., quantum electrodynamics or the Yukawa
theory, such as the absence of a gap in the massless case, the appearance of
a resonance, and the ultraviolet divergence, which can be studied with
mathematical rigor without being obstructed by additional complications, such
as dispersion of the sources or additionally spin degrees of freedom of the
fields. In the case of a massless scalar field, the Spin-Boson model   describes a two-level  atom   that interacts with a photon field and is
therefore frequently employed in quantum optics.  Furthermore, in the massive
case, the corresponding interaction is similar to the one in the Yukawa theory.
The unperturbed energies of the two-level system shall be denoted by real
numbers $0 = e_0<e_1$. It is well-known that after switching on the interaction
with the second-quantized scalar field that may induce transitions between the
two levels, the free ground state energy $e_0$ is shifted to the interacting
ground state energy $e_0>\lambda_0\in\mathbb R$ on the real line while the free
excited state with energy $e_1$ turns into a resonance with an ``energy''
$\lambda_1\in\mathbb C$ situated in the lower complex plane.  In a recent
work~\cite{bdh-scat}, a formula revealing the relation between the resonance
$\lambda_1$ and the integral kernel of the scattering matrix was derived for
the case of a massless scalar field. It was proven that the scattering matrix
coefficients of one-boson scattering processes, excluding forward scattering,
feature the expected Lorenzian shape in leading order in the neighborhood of
the real part of the resonance $\lambda_1$. More precisely, it was shown that
the leading order in the coupling constant $g$ (for small $g$) of the integral
kernel of the transition matrix $T$ fulfills 
\begin{align} \label{eq:transition-matrix-massless} 
    T(k,k')
    \sim 4\pi i g^2 \norm{\Psi_{\lambda_0}}^{ -2 } f(k)^2
    \delta(|k|-|k'|)
    \frac{\Re\lambda_1 - \lambda_0}
         {(|k|+\lambda_0 -  \lambda_1)
         (|k|-\lambda_0  +  \overline{\lambda_1})}.
\end{align} 
Here, $\Psi_{\lambda_0}$ denotes the (due to the construction, unnormalized)
ground state corresponding to $\lambda_0$ and $\delta$ the Dirac delta
distribution. Due to the absence of a spectral gap, a subtle study by means of
multi-scale perturbation analysis was necessary to construct the ground state
and resonance and control the required
spectral estimates \cite{bdh-res}. To the best of our knowledge this is one of
the first results towards a clarification of the relation between   resonances   and
scattering theory in quantum field theory in the same vain as it was done in
quantum mechanics, see~\cite{simonnbody} and references therein. In contrast,
it has to be emphasized  that the relation between the imaginary value
of the resonance and the decay rate of the unstable excited state has been
established rigorously in various models of quantum field theory in several
articles \cite{sigal,hasler,sf,bmw}.  The result in \cite{bdh-scat} and also
the one provided here, hence, naturally draw from many existing results:
Resonance theory for models of quantum field theory has been developed in many
works mainly studying the massless case of various   models  of quantum field
theory with methods of renormalization group, see, e.g.,
\cite{bfs1,bfs2,bfs3,bcfs,bfs100,bbf,fgs100,feshbach,s,f,bffs}, as well as with
methods of multi-scale perturbation analysis, see, e.g.,
\cite{pizzo1,pizzo2,bach,bbp}.  Scattering theory has also been developed for
various models of quantum field theory, see, e.g.,
\cite{fau1,fau2,fau3,fgs1,fgs2}, and in particular for the massless Spin-Boson
model, see, e.g., \cite{rgk,rk,rk2,derezinski,bkz}.  

In the previous work \cite{bdh-scat}, the main tool used to control the time
evolution in the scattering regime, and hence, the scattering matrix
coefficients, was the Laplace transform representation of the unitary time
evolution generated by  the  corresponding Hamiltonian $H$, i.e.,
\begin{align} 
    \left\langle \phi , e^{-itH} \psi \right\rangle 
    = 
    \lim_{\epsilon\downarrow 0}\frac{1}{2\pi i}\int_{\mathbb
    R+i\epsilon}\mathrm{d}z\, e^{-itz}
    \left\langle \phi , \left( H-z \right)^{-1} \psi \right\rangle.
    \label{eq:laplace} 
\end{align}
In oder to justify this identity in a rigorous sense, precise control of the
resolvent close to the real axis is needed to infer sufficient decay for the
integral to converge. For this purpose, the Hamiltonian was studied with the
help of a conveniently chosen complex dilation in which it exhibits a spectrum
consisting of the ground state energy $\lambda_0$, a resonance $\lambda_1$
having negative imaginary part, and  the rest of the spectrum being  
  localized in cones in
the lower complex plane attached to $\lambda_0$ and $\lambda_1$, respectively.
Thanks to this fact,  a well-defined meaning can be given to $ \eqref{eq:laplace}  $   by
deforming the integration   contour  $\mathbb R + i\epsilon$ at $-\infty$ and
$+\infty$ towards the lower complex plane.

In the case of a scalar field with mass $m>0$ as discussed in this work, this
strategy fails. The reason is that  the spectrum  of  the corresponding  dilated  unperturbed Hamiltonian
   contains the points    
\begin{align}
    \label{eq:spectrum_free}
      \{ e_0  + k m  \}_{k \in \N_0}\cup  \{ e_1  + k m  \}_{k \in \N_0}   ,  \qquad \text{where}  \qquad \N_0:= \N \cup \{0\}  .
\end{align}
This leads to an absence of
    decay of the corresponding complex dilated resolvent close to the real
    line, which, in~\cite{bdh-scat}, was a crucial ingredient to control the
    time evolution in the scattering regime.
 Therefore, compared to
\cite{bdh-scat}, a different strategy to control the time evolution has to be
developed which is the content of this paper. 
As discussed in Section \ref{sec:mainresult} below, we use Mourre theory to obtain the required spectral control.  In particular, we combine Mourre theory with perturbation theory and the Feshbach-Schur map. In Section \ref{sec:mainresult} we compare this approach to the method of complex dilation which was employed in \cite{bdh-scat,oy}.

We point out to the reader that, in general, Mourre theory has been studied in a variety of models (see, e.g., \cite{bookabg,amrein,cycon,ggm}). We emphasize, however, that our application of this theory is non-standard. In the spirit of \cite{ahs,fms}, we prove a ``reduced'' limiting absorption principle for the unperturbed Hamiltonian at the excited energy $e_1$ and we apply perturbation theory -- see Lemma~\ref{keyest1} and Proposition~\ref{prop:absenceofacspecclose} (iii). 
  One of the main achievements of the present paper is then to combine the obtained limiting absorption principle with a suitable application of the Feshbach-Schur map. Using in addition Fermi's Golden Rule, we then manage to obtain the required control of the time evolution.       
\\

The paper is structured as follows: In Section~\ref{sec:defmodel} we define the
massive Spin-Boson model and recall its properties relevant to this work, in
Sections~\ref{sec:scattering} and \ref{sec:gs} we review the required results
from scattering theory and the constructions of the ground state,
and in Section~\ref{sec:mainresult} we present   our  main result, i.e.,
Theorem~\ref{FKcor}. The remaining sections consist of the main technical
ingredient given in Section~\ref{sec:mourre}  and its proof in Section~\ref{app:limab}, the proof of the main result in
Section~\ref{sec:proof-mainresult}, and an Appendix for the reasons of
self-containedness. We lay out a roadmap for these sections in the end of
Section~\ref{sec:mainresult}.

\subsection{Definition of the Spin-Boson model}
\label{sec:defmodel}
In  this section  we introduce the considered model and
 state preliminary definitions and well-known tools and facts from which we
start our analysis.  
Most parts of this section are drawn from \cite[Section 1.1]{bdh-scat}. If the reader is already familiar with \cite{bdh-scat}, this section can be skipped -- except for Assumption \ref{as}. 
 \\

The non-interacting Spin-Boson Hamiltonian is defined as
\begin{align}
\label{h0def}
H_0    \equiv H_0(\omega)  :=K + H_f , \qquad K:= \begin{pmatrix}
e_1 & 0 \\
0 & e_0
\end{pmatrix} ,
\qquad
H_f   \equiv  H_f( \omega)  :=\int \mathrm{d^3}k \, \omega(k) a(k)^* a(k).
\end{align}
We regard $K$  as an idealized free Hamiltonian of a two-level atom, where $0 = e_0 <e_1$ denote its two energy levels.
Moreover, the annihilation and creation operators $a,a^*$ are defined on the standard Fock space in \eqref{eq:CCRdef} below and $H_f$ is the  free
Hamiltonian of  scalar field    having dispersion relation
$\omega(k)=\sqrt{k^2+m^2}$. In this work we only consider massive scalar fields, i.e., $m>0$.
In the remainder of this work we sometimes refer to $K$ as the atomic part and to $H_f$ as the free field part of the Hamiltonian. Furthermore, the sum of those operators, $H_0$, we simply call ``free Hamiltonian''. 
The interaction term reads
\begin{align}
\label{def:VPHI}
V   \equiv V(f)   := \sigma_1\otimes \Phi(f), \quad \text{where} \quad  \sigma_1:= \begin{pmatrix}
0 & 1\\ 1&0
\end{pmatrix}  \quad \text{and}  \quad \Phi(f):= a(f) + a(f)^* .
\end{align}
We point out to the reader that our proofs 
require that the boson form factor $f$ satisfies $f, Df, D^2 f \in L^2( \mathbb{R}^3 )$, where $D$ is the generator of dilations introduced in Definition  \ref{def:secquant} (ii) below.     We also suppose that $f$ is spherically symmetric and  use this in order to simplify our notation (a minor modification in some of our calculations would be necessary in order to drop this assumption).  We identify $ f(k) \equiv f(|k|) $ and assume that  
\begin{align}\label{fmaszero}
f(\sqrt{e_1^2-m^2})>0. 
\end{align}   
In particular, $f$ does not have to be analytic and the infrared singularity is not an issue here.  
This being said, for concreteness, we consider a particular choice that meets the conditions above in the remainder of the paper, i.e.,
\begin{align}
f: \R^3 \setminus \{0\}\to \R , \qquad k\mapsto e^{-k^2/\Lambda^2}\omega(k)^{-\frac{1}{2}} .
\label{eq:f}
\end{align}
The relativistic form factor of a scalar field would be  $ f(k)=(2\pi)^{-\frac{3}{2}}(2\omega(k))^{-\frac{1}{2}} $.  Such an $f$, however, is not square integrable, and therefore, renders the model ill-defined. This is referred to as  ultraviolet divergence.  In our case, the gaussian factor in \eqref{eq:f} acts as an ultraviolet cut-off for $\Lambda>0$ being the ultraviolet cut-off parameter.  For the sake of simplicity, we absorb the missing factor  $2^{-\frac{1}{2}}(2\pi)^{-\frac{3}{2}}$  in the coupling constant $g$.

The full Spin-Boson Hamiltonian is then defined as
\begin{align}
\label{eq:H}
H  \equiv  H(\omega, f)    :=   H_0(\omega) + g V(f) \equiv  H_0 + g V  
\end{align}
for some  coupling constant $g  \geq   0$, on the
Hilbert space
\begin{align}
\mathcal H := \mathcal K \otimes \mathcal F\left[ \mathfrak{h}\right] , \qquad
\mathcal K:= \C^2, 
\end{align}
where 
\begin{align}
\mathcal F\left[ \mathfrak{h}\right] :=  \bigoplus^\infty_{n=0} \mathcal F_n\left[ \mathfrak{h}\right] 
,\qquad
\mathcal F_n\left[ \mathfrak{h}\right] := 
\mathfrak{h}^{\odot n},\qquad 
\mathfrak h:= L^2(\mathbb R^3,\C)
\end{align}
denotes the standard bosonic Fock space, and  the   superscript $\odot n$ denotes the
n-th symmetric tensor product,  where by convention $\mathfrak{h}^{\odot 0}\equiv
\C$. Note that we identify $K\equiv K\otimes 1_{\mathcal F[\mathfrak h]}$ and
$H_f\equiv 1_{\mathcal K}\otimes H_f$ in our notation (see Remark \ref{R} below).  

Due to the direct sum, an element $\Psi \in
\mathcal F[\mathfrak{h}]$
can be represented as a family $(\psi^{ (n)})_{n\in\N_0}$   of wave functions $\psi^{ (n)} \in \mathfrak{h}^{\odot n}$  where we recall  $\mathbb{N}_0 =  \mathbb{N} \cup \{ 0 \}$.  The state $\Psi$ with $\psi^{ (0)}=1$ and $\psi^{ (n)}=0$ for all $n\geq 1$ is called the vacuum and is denoted by
\begin{align}
\label{Omega}
\Omega:=(1,0,0,\dots)\in \mathcal F\left[ \mathfrak{h}\right] .
\end{align}
For any $h\in \mathfrak{h}$ and $\Psi=(\psi^{ (n)})_{n\in\N_0} \in \mathcal F[\mathfrak{h}]$, such as the vector below belongs  to  $\mathcal F[\mathfrak{h}]$, we  define the creation operator 
\begin{align}
\left( a(h)^*\Psi \right)_{n\in\N_0}:=\left( 0,h\odot \psi^{ (0)}  , \sqrt 2 h\odot
\psi^{ (1)}, \dots \right) , \quad \left( a(h)^*\Psi \right)^{ (n)}=\sqrt{n} h\odot
\psi^{ (n-1)}
\label{eq:CCRdef}
\end{align}
and the annihilation operator $a(h)$ as the respective adjoint. Occasionally,
we shall also use the physics notation 
\begin{align}
    a(h)^*=\int \mathrm{d^3}k \, h(k) a(k)^*,
\end{align}
where the action of these operators  in the $n$ boson sector of a
vector $\Psi= (\psi^{ (n)})_{n\in\N_0} \in\mathcal F[\mathfrak{h}]$ is to be understood  (only formally)  as:
\begin{align}
    \label{eq:aformal}
    \left(a(k) \Psi  \right)^{ (n)}(k_1,...,k_n)&=\sqrt{n+1}
    \psi^{ (n+1)}(k,k_1,...,k_n),  \\
    \left(a(k)^* \Psi  \right)^{ (n)} (k_1,...,k_n)&=\frac{1}{\sqrt{n}}\sum^n_{i=1}
    \delta^{(3)}(k-k_i) \psi^{ (n-1)}(k_1,...,\tilde k_i,...,k_n)  . 
\end{align}
Here,  the notation  $\tilde \cdot $ means that the corresponding variable is
omitted and $\delta$ denotes the Dirac's delta distribution acting on Schwartz
test functions. Note that $a$ and $a^*$ fulfill the canonical commutation
relations:
\begin{align}
\label{eq:ccr}
  \forall h,l\in\mathfrak{h}, \qquad \left[a(h),a(l)^*   \right]=\left\langle h, l\right\rangle_2 , \qquad \left[a (h),a(l)   \right]=0 , \qquad \left[a(h)^*,a(l)^*   \right]=0.
\end{align}
Throughout this paper we address the case of small coupling, i.e., we assume the coupling constant g to be sufficiently small. We do this only a finite number of times which assures that there is a $\boldsymbol g>0$ such that all results hold true for coupling constants $0<g<\boldsymbol g$.

As mentioned   above, we consider a two-level system with two distinct energy levels at $0=e_0<m<e_1$. Moreover, suppose that the following assumption holds true: 
\begin{assumption}
\label{as}
We suppose that $e_1-e_0 \notin m\N$. This  implies  
\begin{align}
\delta:= \mathrm{dist }(e_1-e_0, m\N) >0, 
\end{align}
where the symbol $\mathrm{dist }$  stands for  the Euclidean distance.    Moreover,  we assume the mass of the scalar field to be smaller than the  energy level $e_1$ in order to allow for scattering processes. 
\end{assumption} 
Speaking in physical terms, this assumption excludes the possibility that a certain  number of photons with zero momentum are able to flip the atom to the excited state.

Let us recall some well-known facts about the introduced model. 
Clearly, $K$ is self-adjoint on $\mathcal K$ and its spectrum 
consists of two eigenvalues $e_0$ and $e_1$. The corresponding eigenvectors are
\begin{align}
\label{varphi}
\varphi_0= \left( 0,1 \right)^T \qquad \text{and} \qquad  \varphi_1= \left(1,0\right)^T \qquad \text{with} \qquad K \varphi_i =e_i \varphi_i , \quad i=0,1.
\end{align}
Moreover, $H_f$ is self-adjoint on its natural domain $\mathcal D(H_f)\subset \mathcal F[\mathfrak{h}]$ and its spectrum  is given by   $\sigma (H_f)= \{0\}\cup [m, \infty )$.  Consequently, the spectrum of $H_0$ is given by
$\sigma (H_0)=  \{e_0\} \cup[e_0+m, \infty )$,  $e_0,e_1$ are eigenvalues of $H_0$  and, asumming that  $e_0 +m < e_1$,     the later is embedded in the absolutely continuous part of the spectrum of $H_0$ (see \cite{reedsimon1}).

Finally, also the self-adjointness of the full Hamiltonian $H$ is well-known
 (see,  e.g.,  \cite{bdh-scat} and \cite{spohnspin}, see also  \cite[Lemma 21]{fgs3}).  
\begin{proposition}
\label{PTOPNDJO}
 For every  $h\in \mathfrak{h}$ and $ a(h)^{\#}  \in \{ a(h)^{*}, a(h) \} $,
\begin{align}
\norm{a(h)^{\#} (H_f+1)^{-\frac{1}{2}}}
&\leq 
C  \norm{h}_2  ,
\end{align} 
where $C$ is a positive constant. 
This implies that $gV$ is infinitesimally bounded with respect to $H_0$ and,
consequently, $H$ is self-adjoint and bounded from below, on the domain 
\begin{align}
\mathcal D(H) = D(H_0)= \mathcal D( \mathbbm 1_{ K} \otimes H_f ),
\end{align}
and the operators
    \begin{align}
       & H_f(H+i)^{-1},  
      \qquad
        H(H_f+1)^{-1} 
    \end{align}
are bounded.     
\end{proposition}
\begin{remark}\label{R} 
    In  this  work we omit
    spelling out identities whenever unambiguous. 
    For every vector spaces $V_1$,  $V_2$ and
    operators $ A_1 $ and $A_2$ defined on $V_1$ and $V_2$, respectively, we
    identify \begin{equation}\label{iden} A_1 \equiv A_1 \otimes \mathbbm
        1_{V_2}, \hspace{2cm}  A_2  \equiv \mathbbm 1_{V_1} \otimes A_2 .
    \end{equation}
    In order to simplify our notation further, and whenever
    unambiguous, we do not utilize specific  notations for every inner product
    or norm that we  employ.    
\end{remark}   

\subsection{Ground state}
\label{sec:gs}
The existence of a unique ground state has already been proven in the more complicated situation of a massless scalar field; see e.g.\ \cite{spohnspin} and \cite{bdh-res} and for the massive model at stake it can be shown using  regular  perturbation theory. However, for the convenience of the reader, we provide a detailed proof in Appendix \ref{app:gs-proof}. 
\begin{proposition}[Ground state]
\label{prop:gs}
 For any    $ g \geq 0$,  $H$ has a unique ground state, i.e., $\lambda_0 = \inf \sigma( H )$ is a simple eigenvalue of $H$.
We have
\begin{align}
\lambda_0=e_0-g^2\Gamma_0+R_0(g),  \qquad \text{where} \qquad \Gamma_0:= \norm{f/(e_1-e_0+\omega)}^2 ,
\end{align}
and there is a constant $C>0$ such that $|R_0(g)|\leq Cg^4$. 
Furthermore, denoting by $\Psi_{\lambda_0}$ the (unnormalized) ground state constructed in Appendix \ref{app:gs-proof}, we have that
\begin{align}\label{eq:approx_GS_intro}
\norm{\Psi_{\lambda_0}-\varphi_0\otimes\Omega }\leq Cg .
\end{align}
\end{proposition} 
 The existence of a ground state can be established for any value of $g$, see \cite{derezinski}.

\subsection{Scattering theory}
\label{sec:scattering}
Finally, we give a short review of scattering theory,  in models of quantum field theory,  which will be necessary to
state the main results  in
Section~\ref{sec:mainresult}.
 For a more detailed introduction we refer to \cite[ Section 1.2]{bdh-scat}.
\begin{definition}[Basic components of scattering theory]
\label{defasymptop}
We denote the dense subspace of compactly supported, smooth, and complex-valued
functions on $\R^3\setminus \{0 \}$   by
\begin{align}
    \label{def:h0}
    \mathfrak{h}_0 :=\mathit C_c^\infty (\R^3\setminus \{0 \},\C) \subset
    \mathfrak{h}.
\end{align}
Furthermore, we define the following objects:
\begin{enumerate}
    \item[(i)] For $h\in\mathfrak{h}_0$, the limit operators
\begin{align}
    \label{asymptop}
    a_\pm(h)\Psi :=\lim\limits_{t\to\pm \infty}a_t(h)\Psi, \quad
    a_t(h):=e^{itH}a(h_t) e^{-itH},
    \quad
  h_t(k):=h(k) e^{ - it\omega(k)} ,
\end{align}
for all $\Psi\in\mathcal H$ such that the limit exists, and also their
respective adjoints $a_\pm^*(h)$.
 \item[(ii)] The two-body scattering matrix coefficients:
\begin{align}
\label{eq:2bodyscat}
S(h,l)= \norm{\Psi_{\lambda_0}}^{-2}\left\langle
a_+(h)^*\Psi_{\lambda_0},a_-(l)^* \Psi_{\lambda_0} \right\rangle, \qquad \forall
h,l\in \mathfrak h_0 ,
\end{align}
where the factor $\norm{\Psi_{\lambda_0}}^{-2}$ appears due to the fact that,
  in our notation, the ground state
$\Psi_{\lambda_0}$ is not necessarily normalized.
 \item[(iii)] The two-body transition matrix
coefficients given by
\begin{align}
T(h,l)=S(h,l) -\left\langle h , l \right\rangle_2 \qquad \forall h,l\in \mathfrak h_0 .
\label{eq:Tmatrix}
\end{align}
\end{enumerate}
\end{definition} 
The  operators $a_\pm$ and $a_\pm^*$ are called asymptotic
outgoing/incoming annihilation and creation operators. 
For  $\Psi\in\mathcal D(H_0^{1/2})$, the limits  \eqref{asymptop} exist.  The proof of this is obtained from the fundamental theorem of calculus, i.e. we write (for example for $a_{-}$)
     \begin{align} \label{amenosintegral}
            a_-(h)\Psi=a(h)\Psi + ig\int^0_{-\infty} \mathrm{d}s\, 
        e^{isH}
        \langle h_s,f\rangle_2\,
            \sigma_1
            e^{-isH} \Psi,
        \end{align} 
and we apply integration by parts to show that the integral above exists  
 (see, e.g.,
\cite{fau1,fau2,fau3,fgs1,fgs2,rgk,rk,rk2,derezinski,bkz}).  All details of this proof are presented in \cite{bdh-scat} for the massless case.  
Therein, also other useful results are shown (see \cite[Lemma 4.1]{bdh-scat}), e.g., 
\begin{align}\label{amenospsi}
a_\pm(h)\Psi_{\lambda_0}=0.
\end{align}

 The starting point for the analysis of the transition matrix  is what we call preliminary scattering formula: for $h,l\in \mathfrak{h}_0$,
   \begin{align}
T(h,l)= -2\pi ig \norm{\Psi_{\lambda_0}}^{-2}\left\langle 
 \sigma_1  \Psi_{\lambda_0}, a_-(W)^* \Psi_{\lambda_0} \right\rangle, \hspace{.2cm}  W( k):=|k|^2 l(k) \int\mathrm{d}\Sigma \, \overline{h(|k|,\Sigma)}f(|k|,\Sigma),  
 \label{W} 
\end{align}
where we use spherical coordinates $k=(|k|,\Sigma)$. Eq.\  \eqref{W} is proven for the massless case in \cite[Theorem 4.3]{bdh-scat}. 
In our setting, when considering massive scalar fields, the proofs of \eqref{amenosintegral}, \eqref{amenospsi} and \eqref{W} follow the same line of arguments as the proofs of Lemma 4.1 and Theorem 4.3 in \cite{bdh-scat}, and therefore, we do not repeat them here.

The matrix coefficients $S(h,l)$  can be interpreted as transition
amplitudes of the scattering process for the following scenario: One  incoming boson with wave
function $l$ is scattered at the two-level atom into an outgoing boson with
wave function $h$. We point out to the reader that in this work we focus on
one-photon processes only, however, the matrix coefficients of
multi-photon processes can be defined likewise.

\section{Main result}
\label{sec:mainresult}

 We now come to our main result, Theorem~\ref{FKcor} below, which makes
precise the relation between the scattering matrix kernel and the resonance.
\begin{definition}
     Using the notation $d^3x\equiv d\Sigma r^2 dr$ for solid angles
    $\Sigma$ and radius $r$ in spherical coordinates, we define, for all
    $h,l\in\mathfrak{h}_0$,
    \begin{align}
        \label{eq:G-def}
        G_{h,l}: \R \to \C , \qquad r \mapsto G_{h,l}(r):=
        \begin{cases}
            \int {d}\Sigma {d}\Sigma' \,  r^4
            \overline{h(r,\Sigma)} l(r,\Sigma') f(r)^2  \qquad &\text{for}
            \quad r\geq 0 
            \\ 
            0 \quad &\text{for} \quad r<0 
        \end{cases}. 
    \end{align}
    In the proofs below we will drop the indices $h,l$ and write $G_{h,l}\equiv
    G$.
\end{definition}  
\begin{theorem}[Scattering formula]
    \label{FKcor}
    Suppose that Assumption \ref{as} holds.  There exists a complex
    number $\Gamma_{-0}$ with $ \Im   \Gamma_{-0} > 0  $  such
    that for all $h,l\in\mathfrak{h}_0$ and $g>0$  sufficiently small,
    the transition matrix coefficients \eqref{eq:Tmatrix} are given by
    \begin{align} 
        \label{scatteringformulapp}
        T(h,l)= & T_{P}(h,l) + R(h,l) ,
    \end{align} 
    where  
    \begin{align} 
        \label{scatteringformulapp1} 
        T_{P}(h,l):= & 4 \pi i g^2
        \norm{\Psi_{\lambda_0}}^{-2}  \int \mathrm{d} r \,    \frac{G_{h,l}(r)
        \left( e_1-g^2 \Re  \Gamma_{- 0}   -\lambda_0 \right)}{
        \left( \omega(r)+  \lambda_0- \big (  e_1 -g^2 \Gamma_{- 0}  \big )
        \right)\left(\omega(r)-  \lambda_0+  \big (  e_1 -g^2 \overline{ \Gamma_{- 0}}
        \big ) \right)}  ,
    \end{align} 
    and there is a constant $C(h,l)>0$ such that 
    \begin{align}  
        \label{scatteringkernel}
        |R(h,l) |\leq C(h,l)  g^{2} g^{1/3} |\log(g)|   .
    \end{align} 
    In \eqref{LL2} below we give an explicit expression of $\Gamma_{-0}$.
\end{theorem}
 Not surprisingly, it turns out that $  \widetilde \lambda_1 := e_1  - g^2
\Gamma_{-0} $ is the leading term of the resonance,  up to order
$g^2$.  This connection can be made by the standard construction of the
resonance by means of complex dilation. This computation is not carried
out here 
since we wanted to focus on the methods of Mourre theory rather than complex dilation; 
see,  e.g., \cite{bbp}  for such a
construction for  massless  fields using the method of complex dilation. Note that, in our situation, the construction is much easier since the dilated Hamiltonian exhibits spectral gaps. For treating resonances within the realm of Mourre theory we refer to  \cite{km,kms,cgh,fms}.

In order to compare this formula with the massless case, see
\eqref{eq:transition-matrix-massless}, we may rewrite
\eqref{scatteringformulapp1} in integral kernel form which takes the
form 
\begin{align} 
    \label{eq:transition-matrix-massive}
    T(k,k')
    \sim 4\pi i g^2 \norm{\Psi_{\lambda_0}}^{ -2  } f(k)^2
    \, \frac{|k|\delta(\omega(k)-\omega(k'))}{\omega(k)} \,
    \frac{\Re\widetilde\lambda_1 - \lambda_0}
         {(|k|+\lambda_0 - \widetilde\lambda_1)
         (|k|-\lambda_0 +  \overline{\widetilde\lambda_1})}.
\end{align}
There are only two differences in the formulas
\eqref{eq:transition-matrix-massive} and \eqref{eq:transition-matrix-massless}.
One is due to the different dispersion relations $\omega(k)=\sqrt{|k|^2+m^2}$
and $\omega(k)=|k|$ for the massive and massless case, respectively, and the
other due to the fact that, in \eqref{eq:transition-matrix-massless},
$\lambda_1$ figures the non-perturbative resonance  while, in
\eqref{eq:transition-matrix-massive}, the entity $\widetilde\lambda_1$ is only the
 second order perturbation in $g$ for small $g$ as explained above.
However, the latter difference is  not relevant  as the rest term $R(h,l)$ in
both     cases  is of  order  $g^2 \,g^{1/3}|\log g|$, and thus, will swallow this difference
anyway. 

The difference in the order in $g$ of the given estimates of the rest terms
$R(h,l)$ between the massive, i.e., $g^2 \,g^{1/3}|\log g|$ in Theorem
\ref{FKcor}, and  the   massless case, i.e., $g^2\, g|\log g|$ in Theorem 2.2 in
\cite{bdh-scat},    is  solely due to the different techniques which were
employed. While in this paper the required spectral information was inferred by
Mourre theory in the paper \cite{bdh-scat} the method of complex dilation was
used. If a fair comparison of both techniques is possible at all, from our
experience, it turns out that Mourre theory requires less information about the
model, especially, no analyticity properties, to start with, however, gives a
little more imprecise estimates of the remainders. In turn, the method of
complex dilation is based on these analyticity properties but, given this
information,  one   is able to produce slightly better estimates on the remainders.   Since the model features a scalar interaction, the
physical perturbation processes only differ for even orders in $g$. Hence, the
different estimates of the remainders inferred by our application of Mourre
theory and   the  method of complex dilations can be expected to be physically
insignificant. Furthermore, also technically, there seems to be room for
improvement.  \\

 Compared to our previous derivation of the transition matrix formula
\eqref{eq:transition-matrix-massless}, see \cite{bdh-scat}, for the massless
Spin-Boson model, there are two main innovations in the strategy of proof.
First, as already explained, we do not rely on complex dilations anymore but
instead use Mourre theory to infer the required spectral information. And
second, as mentioned already in the introduction, we handle the problem caused
by the nature of the spectrum of the free  dilated  Hamiltonian 
 (see \eqref{eq:spectrum_free}),
which is    a complication   due to non-zero boson mass. In previous works
\cite{bdh-scat} and \cite{bdh-res}, complex dilations  were  used both for the
construction of the resonance as well as the control of required spectral
properties, in particular, the estimates on the relevant resolvents. 

Formally, the main steps of our proof of Theorem~\ref{FKcor} are the following. First, after some computations, we arrive at the formula $T(h,l) = 2 \pi   \norm{\Psi_{\lambda_0}}^{-2} (  T^{(1)}- T^{(2)} )$, where
\begin{align*}
T^{(1)} &   =  g^2\int_0^\infty  \mathrm{d}t  \, \zeta(t)  \left\langle   \sigma_1 \Psi_{\lambda_0}, e^{-itH} \sigma_1\Psi_{\lambda_0}\right\rangle , \\
T^{(2)} &= g^2  \int_0^\infty  \mathrm{d}t  \int_0^\infty \mathrm{d} r\,   G(r) e^{it(\omega(r)-\lambda_0)}     \left\langle \sigma_1  \Psi_{\lambda_0}, e^{itH}\sigma_1\Psi_{\lambda_0}\right\rangle ,
\end{align*}
and $\zeta$, $G$ are some functions defined in \eqref{def:W1sttt} below. Next, we study the quantity $\langle \sigma_1  \Psi_{\lambda_0}, e^{\pm itH}\sigma_1\Psi_{\lambda_0} \rangle$. Using \eqref{eq:approx_GS_intro} and the Spectral Theorem, we rewrite
\begin{align}
\left\langle \sigma_1  \Psi_{\lambda_0}, e^{\pm itH}\sigma_1\Psi_{\lambda_0}\right\rangle = \pi^{-1}\lim\limits_{\epsilon\to 0^+}\int_\R \mathrm{d}r \, \chi(r) e^{-itr}  \Im  \left\langle  \Phi_1,  (H -  r \pm i\epsilon)^{-1}\Phi_1\right\rangle , \label{eq:estim_intro}
\end{align}
with $\Phi_1 = \varphi_1 \otimes \Omega$, $\varphi_1$ is an eigenstate of $K$ associated to the excited energy $e_1$ (see \eqref{varphi}), and $\chi$ is a smooth function supported in a suitable small interval containing $e_1$. Using the Feshbach-Schur map, we write
\begin{align*}
 \left\langle  \Phi_1,  (H -  r \pm i\epsilon)^{-1}\Phi_1\right\rangle = \left\langle  \Phi_1,  F_P(H-r\pm i \epsilon)^{-1} \Phi_1\right\rangle ,
\end{align*}
where $F_P(H-z):=P(H-z)P -g^2 PV\overline P (H_{\overline{P}}-z)^{-1}\overline P VP$, with $P$ the orthogonal projection onto the vector space spanned by $\Phi_1$ and $H_{\overline{P}} = \overline P H \overline P$. The key point then consists in computing an expansion of $F_P(H-r\pm i 0^+ )^{-1}$ in $g$ and $r$ by means of Mourre's theory. We establish regularity properties of the boundary values of the resolvent of $H_{\overline{P}}$ near the real axis (see \eqref{holder} below) using in particular arguments of \cite{fms}, from which we deduce that $F_P(H-r\pm i 0^+ )^{-1}$ is H{\"o}lder-continuous of order $\frac12$ in \emph{both} $r$ and $g$. Together with suitable estimates on remainder terms, this allows us to prove Theorem~\ref{FKcor}.

To our knowledge, our approach to study \eqref{eq:estim_intro} has not been used previously in the literature (see although \cite{sigal,bfs100,hasler,km} for other expansions, with different purposes, of quantities similar to \eqref{eq:estim_intro}). We believe that our argument may find applications in other contexts. \\

The  main technical import for
the proof of Theorem~\ref{FKcor} is contained in the next
Section~\ref{sec:limab}. There, we provide a central Mourre estimate in
Lemma~\ref{lemma:comestclose} which implies a limiting absorption principle  in 
Proposition~\ref{prop:absenceofacspecclose}.  The latter is employed in
Section~\ref{sec:timeevo}, in a combination with the Feshbach-Schur map as mentioned above,   to control the time evolution in the scattering
regime, and hence, the transition matrix coefficient under investigation.  
In Section~\ref{app:limab} we provide a
proof of the limiting absorption principle, i.e.,
Proposition~\ref{prop:absenceofacspecclose}, which  in parts is a
self-contained review of results in the literature but also provides a
non-standard result, see \eqref{holder}, which allows to conveniently apply a
limiting absorption principle in the context of perturbation theory.
\begin{remark} 
\label{rem:const}
In the remainder of this work we denote by $C$ any generic, positive (indeterminate) constant which may change from line to line in the computations but does not depend on   $g$ and the parameters   $z,z',\epsilon, \eta, \beta$ introduced below.    
\end{remark}

\section{Technical ingredients}
\label{sec:mourre}

In this section we derive a formula for the leading order term with respect to the coupling constant of a certain matrix element of time-evolution and estimate the error term.  We rely on two main ingredients,  namely, a limiting absorption principle derived from a Mourre estimate and a Feshbach-Schur map. In the first part, Section \ref{sec:limab}, we introduce  some   notation and   prove  technical lemmas and a Mourre estimate which allows to derive a limiting absorption principle. The latter is also stated     in     this section since we use it as a key tool in order to prove    our   main result. 
  In the second part, Section \ref{sec:timeevo}, we introduce a Feshbach-Schur map and   combine it with   the limiting absorption principle  in order to control a certain matrix element of time-evolution.

\subsection{Limiting absorption principle}
\label{sec:limab}

In this section we present the  limiting absorption principle based on a Mourre estimate for the model at stake. We follow the 
construction of \cite{cycon}, see also  
 \cite{spohnspin,fms,ggm}.  
 We start with introducing    some  notation. 
\begin{definition}
\label{def:secquant}
 Recall that $\mathfrak{h}_0$ has been defined in \eqref{def:h0}. 
\begin{enumerate}
\item[(i)]   
For any  self-adjoint operator $O$,   we define $\mathrm{d}\Gamma(O)  $ as the generator of the unitary one-parameter group $\left\{\Gamma(e^{-itO})\right\}_{t\in\R}$, where 
\begin{align}\label{secondQ}
\Gamma(e^{-itO}) := \bigoplus_{n = 0}^{\infty} (e^{-itO})^{\odot n} ,  \hspace{2cm}    (e^{-itO})^{\odot 0} := 1.    
\end{align} 
 It follows from Stone's theorem that $  \mathrm{d}\Gamma(O) $ is self-adjoint.  Notice that    $H_f=\mathrm{d}\Gamma(\omega)$.
\item[(ii)] 
For $\beta \in \R$, we define the unitary dilation operator
\begin{align}
u_\beta:\mathfrak h \to \mathfrak h, \qquad \varphi(k)\mapsto
\varphi_\beta(k):=e^{\frac{3}{2}\beta}\varphi(e^{\beta}k), \qquad \forall k\in \R^3.
\end{align} 
  We denote by $D$ generator of dilations, which is the  generator of the unitary one-parameter group  $\left\{ u_\beta\right\}_{\beta\in\R}$.      
    Note that $D$ is self-adjoint on $\mathcal D(D)\subset \mathcal H$ due to Stone's theorem.

Moreover,  for $\varphi\in  \mathfrak{h}_0 $ and $\beta \in\R$, we observe  that 
\begin{align}
\label{diffeq21}
\frac{\mathrm{d}}{\mathrm{d} \beta} \varphi_\beta (k) =\frac{1}{2} \left( \nabla_k \cdot k  +  k\cdot \nabla_k    \right)\varphi_\beta (k), \qquad k\in \R^3 .
\end{align}
This implies that the action of $D$ on $ \mathfrak{h}_0 $ is given by  $\frac{i}{2}   (k\cdot \nabla_k + \nabla_k \cdot k)  $. 
\item[(iii)] We     introduce    the function
\begin{align}
\label{def:xi}
\xi :\R^3\to \R , \qquad k\mapsto \xi(k):= k^2/\omega(k) .
\end{align}
\item[(iv)] We    set 
\begin{align}
    \label{eq:H0}
   \mathcal H_0:=\mathcal K\otimes
\mathcal F_{\text{fin}}[\mathfrak h_0] ,
\end{align}
where
\begin{align}
\label{def:denseF}
    \mathcal F_{\text{fin}}[\mathfrak h_0]:= \Big \{ \Psi=(\psi^{ (n)})_{n\in\N_0}\in
        \mathcal F[\mathfrak h] \,\big|\, 
        \exists  N  \in \N_0   : &  
\psi^{ (n)}=  0 \, \forall n\geq N, \\ \notag &  \forall n \in \mathbb{N} :
   \psi^{ (n)} \in C_c^{\infty}(\mathbb{R}^{3n } \setminus \{ 0 \}, \mathbb{C})    \Big \}. 
\end{align}
\item[(v)]  Moreover,  for   every closed operator  $A$,   we denote by 
\begin{align}
\norm{\cdot}_A:=\left(\norm{A \cdot}^2+\norm{\cdot}^2\right)^{1/2},
\end{align}
its graph norm in the domain of $A$.  
\end{enumerate}
\end{definition}
\begin{remark}
\label{rem:densegraph}
Note that $\mathcal H_0$ and $\mathcal F_{\text{fin}}[\mathfrak{h}_0]$  are  dense subsets of  the domains of $H$ and $H_f$ with respect to the graph norm of $H$ and $H_f$, respectively.  In other words,  $\mathcal H_0$ and $\mathcal F_{\text{fin}}[\mathfrak{h}_0]$ are cores of $H$ and $H_f$, respectively. 
\end{remark}
The following statement is a collection of general properties of the objects introduced in Definitions \ref{def:secquant}  which we will use in the remainder of this work.
\begin{lemma}
\label{lemma:funcprop}
The following properties hold true: 
\begin{enumerate}
\item[(i)] $\mathcal F_{\text{fin}}[\mathfrak{h}_0]\subset \mathcal D(H_f)\cap \mathcal D(\mathrm{d}\Gamma(D))$. 
\item[(ii)]      $ \mathcal D(H_f)\subset \mathcal D(\Phi(Df))$ and   $  \Phi(Df)(H_f+1)^{-\frac12}   $ is bounded (recall the definition of $\Phi(f)$ in \eqref{def:VPHI}).    
\item[(iii)]   $ \mathcal D(H_f)\subset \mathcal D(\mathrm{d}\Gamma(\xi)) $  and  $  \mathrm{d}\Gamma(\xi)(H_f+1)^{-1}   $ is bounded.    
\item[(iv)] The operator  $[H_f,  i   \mathrm{d}\Gamma(D)]$  defined as a quadratic form on $\mathcal D(H_f)\cap \mathcal D(\mathrm{d}\Gamma(D))$  can be uniquely extended to a $H$-bounded operator on $\mathcal D(H)=\mathcal D(H_0)$ denoted by $[H_f,  i   \mathrm{d}\Gamma(D)]^0$.  We have the identity:
\begin{align}
[H_f,   i   \mathrm{d}\Gamma(D)]^0= \mathrm{d}\Gamma(\xi) 
\end{align}
 on $\mathcal D(H_0)$.
\item[(v)] The operator  $[\Phi(f),   i \mathrm{d}\Gamma(D)]$  defined as a quadratic form on $\mathcal D(H_f)\cap \mathcal D(\mathrm{d}\Gamma(D))$ can be uniquely extended to a $H$-bounded operator on $\mathcal D(H)=\mathcal D(H_0)$ denoted by $[\Phi(f),   i \mathrm{d}\Gamma(D)]^0$.  We have the identity:
\begin{align}
[\Phi(f),  i \mathrm{d}\Gamma(D)]^0=   \Phi(Df)
\end{align}
 on $\mathcal D(H_0)$.
\end{enumerate}
\end{lemma}
\begin{proof}
\begin{enumerate}
\item[(i)] Clearly, this holds by Definition \ref{def:secquant}.
\item[(ii)]  A direct calculation shows that  $Df \in \mathfrak{h}$.  We conclude the claim  by Proposition \ref{PTOPNDJO}.    
\item[(iii)] Note that, for all $k\in\R^3$,    $\xi(k)=\frac{k^2}{\omega(k)}=\omega(k) \frac{k^2}{k^2+m^2}\leq \omega(k)$. This directly implies the desired result.
\item[(iv)] Clearly,  $[H_f,  i  \mathrm{d}\Gamma(D)]$ can be defined as a quadratic form on  $\mathcal D(H_f)\cap \mathcal D(\mathrm{d}\Gamma(D))$, and hence, it follows from (i) that, for $\psi\in \mathcal F_{\text{fin}}[\mathfrak{h}_0]$, we have
\begin{align}
\left\langle\psi,   [H_f,  i \mathrm{d}\Gamma(D)] \psi \right\rangle=\left\langle\psi,  [\mathrm{d}\Gamma(\omega),  i \mathrm{d}\Gamma(D)]\psi \right\rangle=\left\langle\psi,   \mathrm{d}\Gamma([\omega,  i D])\psi \right\rangle  .
\end{align} 
Moreover, it follows from a direct calculation  that 
\begin{align}
[\omega,  i  D]= \xi  ,
\end{align}
 on $\mathfrak{h}_0$, and hence, 
\begin{align}
\label{eq:densecomm}
\left\langle\psi,   [H_f,   i \mathrm{d}\Gamma(D)] \psi \right\rangle=\left\langle\psi,  \mathrm{d}\Gamma(\xi)\psi \right\rangle  \qquad \forall \psi\in \mathcal F_{\text{fin}}[\mathfrak{h}_0].  
\end{align}
Note that $ \mathcal F_{\text{fin}}[\mathfrak{h}_0]$ is a core of $H_f$.  This together with \eqref{eq:densecomm} and (iii) implies that $[H_f,   i  \mathrm{d}\Gamma(D)] $ uniquely extends to an $H_0$-bounded (and $H$-bounded) operator on $\mathcal D(H)=\mathcal D(H_0)$ denoted by $[H_f,  i \mathrm{d}\Gamma(D)]^0$. 
\item[(v)] This statement follows similarly as (iv) while using (ii) instead of (iii) in the last step.
\end{enumerate}
\end{proof}
For the proof of our main result  it  suffices to control the time evolution only on a spectral subset close to the excited state. In the following we define a
 cut-off  function   with its support localized  in such a subset.  Recall that $\delta > 0$ has been defined in Assumption \ref{as}. 
\begin{definition}
\label{def.chi}
We    fix  $\chi \in C^\infty_c(\R,[0,1])$ such that $\text{\rm{supp} } \chi \subset (e_1-3\delta/4, e_1+3\delta/4)$ and $\chi \big|_{[e_1-\delta/2, e_1+\delta/2]}=1$.
Moreover, for  $0<\kappa < 2$  and $g^2\leq s\leq   g^{\kappa}  $, we define  $\chi_s  $   by $\chi_s(r) : =  \chi( e_1+(r-e_1)/s )$ for all $r\in\R$.
\end{definition}
   Next lemma is proven in Appendix   \ref{app:specproj}.    
\begin{lemma}
\label{lemma:specproj}
For every  $ \upsilon  \in C_c^\infty(\R,[0,1])$, there is a constant $C_\upsilon >0$ such that
\begin{align}
\label{eq:specproj}
\norm{\upsilon(H)-\upsilon(H_0) } \leq gC_{\upsilon} .
\end{align} 
For every $ s > 0 $ in a compact set there is a constant $C$ that depends on this set such that 
\begin{align}
\label{eq:specproj'}
\norm {\chi_s(H)-\chi_s(H_0) } \leq C s^{-1} g  .
\end{align}
 \end{lemma}
In the following  we derive a positive commutator estimate close to the unperturbed eigenvalue  $e_1$.  For this purpose we set (see  \eqref{varphi} and \eqref{eq:H}) 
\begin{align}\label{pl1}
      \hspace{.3cm}  H_{\overline{P}}(\omega, f) & \equiv   H_{\overline{P}}    :=\overline P H (\omega, f) \overline P, 
\hspace{.3cm}  H_{0, \overline{P}}(\omega) \equiv     H_{0, \overline{P}}    :=\overline P H_0 (\omega)  \overline P  \\   \notag \\ \notag 
  H_{f,\overline{P}}(\omega) & \equiv    H_{f,\overline{P}}:=\overline P H_f(\omega) \overline P ,  \hspace{.3cm}   V_{\overline P}(f) \equiv 
 V_{\overline P} := \overline P V(f) \overline P,  \hspace{.3cm} \Phi_{\overline P}(f) = \overline P  \Phi(f) \overline P,    
\end{align}
where,  taking   $  P_{\varphi_1}    $ and $ P_{\Omega}  $  the orthogonal projections on the spans of 
$ \varphi_1 $ and $\Omega$, respectively, we define  
\begin{align}\label{pl2}
P :  = P_{\varphi_1} \otimes P_{\Omega}, \qquad  \overline P=\mathbbm 1_{\mathcal H}-P. 
\end{align}

\begin{remark}\label{DEfiCommut}
It follows from Lemma \ref{lemma:funcprop} that operator $[H_{\overline{P}},  i   \mathrm{d}\Gamma(D)]$, defined as a quadratic form on $\mathcal D(H_f)\cap \mathcal D(\mathrm{d}\Gamma(D))$, can be uniquely extended to a $ H_{\overline{P}} $-bounded operator on $\mathcal D( H_{\overline{P}} )$. We denote this extension by 
\begin{align}\label{Commut}
[ H_{\overline{P}} ,  i   \mathrm{d}\Gamma(D)]^0 =  H_{\overline P}(\xi, Df).
\end{align}  
\end{remark}

\begin{lemma}[Mourre estimate]
\label{lemma:comestclose} 
There  is a  constant $ \alpha > 0$  such that, for sufficiently small  $g>0$, 
\begin{align}
\chi(H_{\overline P}) [H_{\overline P},  i \mathrm{d}\Gamma(D)]^0     \chi(H_{\overline P})\geq \alpha \chi(H_{\overline P})^2  ,
\end{align}
 where we recall Definition \ref{def.chi}. 
\end{lemma} 
\begin{proof}
We take a fixed function 
$ \upsilon  \in C_c^\infty\Big ( (e_1  - \frac{9}{10}\delta, e_1  + \frac{9}{10}\delta  ),[0,1]\Big )$ with $\chi \upsilon = \chi $ (since this is fixed, we identify 
$ C \equiv C_{\upsilon}$ in the constants below).

Note that $\mathrm{d}\Gamma(D)$ commutes  with $\overline P=\mathbbm 1_{\mathcal H}-P$. Then, Lemma \ref{lemma:funcprop} (iv) and (v) yields
\begin{align}
\label{eq:proof_commest}
&\upsilon(H_{\overline P}) [H_{\overline{P}},  i \mathrm{d}\Gamma(D)]^0    \upsilon(H_{\overline P}) = \upsilon (H_{\overline P}) \overline{P} \mathrm{d}\Gamma(\xi)  \overline{P} \upsilon(H_{\overline P}) +g\upsilon (H_{\overline P})  \overline{P} \sigma_1\otimes \Phi(Df)  \overline{P}  \upsilon(H_{\overline P}) .
\end{align}
It follows from Lemma \ref{lemma:funcprop} (ii) that 
 $ \upsilon (H_{\overline P})  \overline{P} \sigma_1\otimes \Phi(Df)   (H_{0, \overline P} + i)^{-1}  \overline{P} (H_{0, \overline P } + i) 
 \upsilon (H_{\overline P})
  $ 
 is bounded (notice that $ (H_{0, \overline P } + i) \upsilon (H_{\overline P}) =  (H_{0, \overline{P}} + i) (H_{\overline P} + i)^{-1}    (H_{\overline P} + i)  \upsilon (H_{\overline P})  $ is bounded,  which follows from  Proposition \ref{PTOPNDJO}). Then, we obtain 
\begin{align}
\label{eq:proof_commest2}
\Big \| g\upsilon(H_{\overline P}) \overline{P} \sigma_1\otimes \Phi(Df)  \overline{P}  \upsilon(H_{\overline P}) \Big  \| 
  \leq  C g .
\end{align}
Similarly as above, we argue that    $\upsilon(H_{\overline P} )    \overline{P} \mathrm{d}\Gamma(\xi) $ and 
 $   \mathrm{d}\Gamma(\xi)  \overline{P} \upsilon (H _{\overline P} ) $  are bounded, using Lemma \ref{lemma:funcprop} (iii). Then,   Lemma  \ref{lemma:specproj} implies that
\begin{align}
\label{eq:proof_commest1}
\upsilon(H_{\overline P}  )    \overline{P} \mathrm{d}\Gamma(\xi)  \overline{P} \upsilon(H_{\overline P}  )  \geq\upsilon(H_{0, \overline{P}}) \overline{P} \mathrm{d}\Gamma(\xi)  \overline{P} \upsilon (H_{0, \overline{P}})  -gC .
\end{align}
Plugging  \eqref{eq:proof_commest1} and  \eqref{eq:proof_commest2} into  \eqref{eq:proof_commest} yields that 
\begin{align}
\label{eq:proof_commest-fin}
\upsilon(H _{\overline P} ) [H_{\overline{P}},  i \mathrm{d}\Gamma(D)]^0    \upsilon ( H_{\overline P}  )
&\geq  \upsilon (H_{0, \overline P}) \overline{P}  \mathrm{d}\Gamma(\xi)  \overline{P} \upsilon(H_{0, \overline P})-g  C  .
\end{align}
Set $\ell \in \mathbb{N}\cup \{ 0 \}$ be such that    
\begin{align}
  e_1  >  \ell m  \hspace{2cm}   e_1  <  (\ell + 1 ) m. 
\end{align}
Notice that  Assumption \ref{as} implies that 
\begin{align}\label{M1}
| e_1 - \ell m |\geq \delta ,
\end{align} 
and  since  
$ \upsilon  \in C_c^\infty\Big ( (e_1  - \frac{9}{10}\delta, e_1  + \frac{9}{10}\delta  ),[0,1]\Big )$, 
\begin{align}\label{M2}
\upsilon (H_{0, \overline{P}})   H_{0, \overline{P}} \upsilon (H_{0, \overline{P}})  \geq \Big ( \ell m  +  \frac{1}{10} \delta \Big ) \upsilon (H_{0, \overline{P}})^2 . 
\end{align}
For any self-adjoint operator $O$, we denote by $E_{O}$  its  resolution of the identity. It follows that
 \begin{align}\label{M2.1}
E_{H_{0, \overline P}} ( U ) =  \begin{cases}   \overline P E_{H_0}(U), & \text{if }  0 \not \in U, \\
 P +     \overline P E_{H_0}(U),    & \text{if }  0 \in   U . \end{cases}
\end{align}
This is a consequence of the fact that the formula in the right  hand side of the  equation above defines a resolution of the identity and the integral of the identity function with respect to it equals $ H_{0, \overline P} $ (notice that $P$ commutes with $E_{H_0}(U)$). Since $0$ does not belong to the support of $ \upsilon $, it follows that 
\begin{align}\label{M2.3}
\upsilon( H_{0, \overline{P}} ) = \upsilon( H_{0 } )\overline{P} = \overline{P} \upsilon( H_{0 } )\overline{P}. 
\end{align}    
Set
$ \mathcal{N} =  \mathrm{d} \Gamma (\mathrm{1})$ the number operator. Since $\omega (k) \geq m$, it follows that 
$  \mathbbm{1}_{ \mathcal{N} >  \ell     } H_{0}   \geq  (\ell + 1) m   $, and therefore (notice that $
\mathcal{N}$ commutes with $ H_{0, \overline{P}} $ and $P$ and recall  \eqref{M2.3}), 
\begin{align}\label{M3}
 m \upsilon (H_{0, \overline{P}})^2  \mathcal{N}= m  \upsilon (H_{0, \overline{P}})^2 \mathbbm{1}_{ \mathcal{N} \leq  \ell     } \mathcal{N} 
 \leq m \ell  \upsilon (H_{0, \overline{P}})^2  .  
\end{align}  
Eqs.\ \eqref{M2} and \eqref{M3} imply that
\begin{align}\label{M4}
\upsilon (H_{0, \overline{P}}) \Big (   H_{0, \overline{P}} - m \mathcal{N} \Big )   \upsilon (H_{0, \overline{P}})   \geq   \frac{1}{10} \delta  \upsilon (H_{0, \overline{P}})^2 . 
\end{align}
Since $\xi(k) = \frac{k^2 + m^2 - m^2}{\omega(k)} = \omega(k) - \frac{m^2}{\omega(k)} \geq \omega(k) - m    $, we get that 
\begin{align}\label{M5}
 \mathrm{d}\Gamma(\xi)  \geq H_{0, \overline{P}} - m \mathcal{N}. 
\end{align}
Eqs.\ \eqref{M4} and \eqref{M5} imply that 
\begin{align}\label{M6}
\upsilon (H_{0, \overline{P}})   \mathrm{d}\Gamma(\xi)    \upsilon (H_{0, \overline{P}})   \geq   \frac{1}{10} \delta  \upsilon (H_{0, \overline{P}})^2.
\end{align}
This together with Lemma \ref{lemma:specproj} 
 and \eqref{eq:proof_commest-fin} lead us to (see also  \eqref{M2.3})
 \begin{align}
\label{eq:proof_commest-finprima}
\upsilon(H _{\overline P} ) [H_{\overline{P}}, i \mathrm{d}\Gamma(D)]^0    \upsilon ( H_{\overline P}  )
&\geq   \frac{1}{10} \delta  \upsilon (H_{ \overline{P}})^2-g  C  .
\end{align}
We multiply by $  \chi(H_{ \overline{P}}) $ from the left and the right and use that $\chi \upsilon = \chi$ to  obtain
 \begin{align}
\label{eq:proof_commest-finprimaprima}
\chi(H _{\overline P} ) [H_{\overline{P}}, i \mathrm{d}\Gamma(D)]^0    \chi ( H_{\overline P}  )
&\geq   \frac{1}{10} \delta  \chi (H_{ \overline{P}})^2-g  C  \chi (H_{ \overline{P}})^2 .
\end{align}
Our desired result follows from  \eqref{eq:proof_commest-finprimaprima}, taking small enough $g$. 
\end{proof}
\begin{proposition}[Limiting absorption principle]
\label{prop:absenceofacspecclose}
We  introduce the notation  
\begin{align}
\label{def:bracket}
\left\langle  \mathrm{d}\Gamma(D)\right\rangle:= \Big ( \big ( \mathrm{d}\Gamma(D) \big )^2 +1 \Big )^{1/2} .
\end{align}  
 For sufficiently small $g>0$,  $\epsilon \in (0,1)$ and $z, z' \in [e_1 - \delta /4, e_1 + \delta /4] $ we have 
\begin{enumerate}
\item[(i)]  $\sigma_\text{pp}(H_{\overline P}) \cap [ e_1 - \delta / 4, e_1 + \delta /4 ] =\emptyset$, where $\sigma_\text{pp}(H_{\overline P})$ denotes  the  pure point spectrum of $H_{\overline P}$.  
\item[(ii)] 
\begin{align}
\label{sup1}
 \norm{\left\langle  \mathrm{d}\Gamma(D)\right\rangle^{-1}   (H_{\overline P}-z\pm i \epsilon)^{-1}    \left\langle  \mathrm{d}\Gamma(D)\right\rangle^{-1}} \leq C , 
\end{align}
and 
\begin{align}
\label{sup2}
\norm{\left\langle   \mathrm{d}\Gamma(D)\right\rangle^{-1}   (H_{0,\overline P}-z\pm i \epsilon)^{-1}    \left\langle \mathrm{d}\Gamma(D)\right\rangle^{-1}} \leq C , 
\end{align}
\end{enumerate}
\item[(iii)] 
\begin{align}
\label{holder}
 \norm{\left\langle  \mathrm{d}\Gamma(D)\right\rangle^{-1} \left(  (H_{\overline P}-z\pm i \epsilon)^{-1}  -(H_{0,\overline P}-z'\pm i \epsilon)^{-1}\right)  \left\langle  \mathrm{d}\Gamma(D)\right\rangle^{-1}} 
\leq C\left( g^{1/2}+|z-z'|^{1/2}   \right) .
\end{align}
 We recall  that the constants above do not depend on $\epsilon ,  z, z' $ and $g$  (c.f.\ Remark \ref{rem:const}). 
\end{proposition}
For the convenience of the reader, we provide a proof  of statements (ii) and (iii)  in Section  \ref{app:limab} - following \cite{cycon}. Notice that statement (iii) is not standard, similar results  are addressed in \cite{fms}.  Their work also draws from \cite{ahs}.  
However, we present no proof for statement (i) since this is not used in the remainder of this work and it is a standard result.

\subsection{Resonance and time evolution}
\label{sec:timeevo}

In this section we introduce a Feshbach-Schur map, c.f.\ \cite{bfs3}, in order to derive a formula for the resolvent restricted to  a spectral subset. This together  with the   limiting absorption principle obtained in Proposition \ref{prop:absenceofacspecclose}  allows then for controlling the leading order term of certain matrix elements of the time evolution (with respect to the coupling constant) and estimate the error term in Lemma \ref{lemma:matrixelement} below.
\begin{definition}
\label{def:feshbach}
 We recall Eqs.\ \eqref{pl1}--\eqref{pl2}.     For all $z\in \C\setminus \sigma(H)$,  we define 
\begin{align}
F_P(z)\equiv F_P(H-z):=P(H-z)P -g^2 PV\overline P (H_{\overline{P}}-z)^{-1}\overline P VP ,
\end{align}
 as an operator on the range of $P$. 
\end{definition} 
The following lemma is an  application of the limiting absorption principle derived in Proposition \ref{prop:absenceofacspecclose} and allows for the control of certain  term of the Feshbach-Schur map introduced in Definition \ref{def:feshbach}.
\begin{lemma}
\label{lemma:mourrefeshbach}  
 For sufficiently small $g$ and every 
$ z \in [ e_1 - \delta/4, e_1 + \delta/4 ] $ and $\epsilon \in (0,1)$, the following estimates hold true:   
\begin{enumerate}
\item[(i)]\begin{align}
\label{eq:mourrefeshbach}
\norm{ PV\overline P (H_{\overline{P}}-z\pm i \epsilon)^{-1}\overline P VP } \leq C .
\end{align}
\item[(ii)] \begin{align}
\label{eq:mourrefeshbach0}
\norm{  PV\overline P (H_{0,\overline{P}}-z\pm i \epsilon)^{-1}\overline P VP } \leq C .
\end{align}
\item[(iii)]  if    $ |z - e_1| \leq r$,  
\begin{align} 
\norm{PV\overline P \left((H_{0,\overline{P}}-e_1\pm i \epsilon)^{-1}  - (H_{\overline{P}}-z\pm i \epsilon)^{-1}\right)\overline P VP}\leq C(g^{1/2} + r^{1/2}) .
\end{align} 
\end{enumerate}
  We recall that the constants $ C $ do not depend on $\epsilon, $ $z$ and $g$  (c.f.\ Remark \ref{rem:const}).  
\end{lemma}
\begin{proof} 
We take $ z \in [ e_1 - \delta/4, e_1 + \delta/4 ] $ and $\epsilon \in (0,1)$.  
Note that $\mathrm{d}\Gamma(D)$ commutes with $P$. Then, it follows from Lemma \ref{lemma:funcprop} (v) together with $\mathrm{d}\Gamma(D)P=0$ that $\mathrm{d}\Gamma(D)\overline P VP=  i  \overline P \sigma_1\otimes a(Df)^*P $, and consequently, 
\begin{align}
\label{boundedop1}
\norm{\mathrm{d}\Gamma(D)\overline P VP}\leq  \norm{a(Df)^*\Omega}
 =  \norm{  Df  } .
\end{align}
 Moreover, we similarly obtain  
\begin{align}
\label{boundedop2}
\norm{ \overline P VP}\leq  C.
\end{align}
We recall the definition of $\left\langle \mathrm{d}\Gamma(D)\right\rangle$ in \eqref{def:bracket} and observe
\begin{align}
\label{boundedop3}
&\norm{\left\langle \mathrm{d}\Gamma(D)\right\rangle \overline P VP}^2= \sup_{\Psi \in \mathcal H , \norm{\Psi}=1} \left\langle \overline P VP\Psi,  \left\langle \mathrm{d}\Gamma(D)\right\rangle^2 \overline P VP\Psi\right\rangle
 \\\notag
&= \sup_{\Psi \in \mathcal H , \norm{\Psi}=1} \left\langle \overline P VP\Psi,  \left(   \mathrm{d}\Gamma(D)^2+1 \right) \overline P VP\Psi\right\rangle
 \leq  \norm{  \mathrm{d}\Gamma(D) \overline P VP}^2 +\norm{ \overline P VP}^2 .
\end{align}
This together with \eqref{boundedop1} and \eqref{boundedop2} implies that  
\begin{align}
\label{proof:dgamma}
\norm{\left\langle \mathrm{d}\Gamma(D)\right\rangle \overline P VP}\leq  C ,
\end{align}
and hence, $\left\langle \mathrm{d}\Gamma(D)\right\rangle \overline P VP$ is a bounded operator on $\mathcal H$. Then, it follows that also its adjoint is a bounded operator.
We obtain that
\begin{align} 
&\norm{PV\overline P (H_{\overline{P}}-z\pm i \epsilon)^{-1}\overline P VP }
\leq  C \norm{\left\langle \mathrm{d}\Gamma(D)\right\rangle^{-1} (H_{\overline{P}}-z \pm i \epsilon)^{-1}\left\langle \mathrm{d}\Gamma(D)\right\rangle^{-1}}  .
\end{align} 
We conclude statement  (i)  by Proposition \ref{prop:absenceofacspecclose} (ii). Statements (ii) and (iii) follow similarly  from Proposition \ref{prop:absenceofacspecclose} (ii) and (iii).
\end{proof} 
Next, we derive an explicit formula for the leading order of the Feshbach-Schur map with respect to the coupling constant. This allows then for an easy approximation of the resolvent restricted on a certain subset in Corollary \ref{Cor} below.
\begin{lemma}
\label{lemma:feshbachest}
For    sufficiently small $r, g>0$, $\epsilon\in (0,1)$ and $z \in \mathbb{R}$ with $ | z - e_1 | \leq r$, we have 
\begin{align}
F_P(H-z\pm i \epsilon)=(e_1-z-g^2 \Gamma_{\pm\epsilon} \pm i \epsilon)P+R_\epsilon(g, r) ,
\end{align}
where   $\norm{R_\epsilon(g, r)}\leq Cg^{2}(g^{1/2} + r^{1/2} ) $   and 
\begin{align}
\label{eq:Gamma}
\Gamma_{\pm\epsilon}: =\int \mathrm{d}^3k\,  \frac{f(k)^2}{\omega(k)-e_1\pm i \epsilon}.
\end{align}
Moreover,   recalling  $m-e_1<0$, we observe that  the limits 
\begin{align}\label{LL1}
\lim_{\epsilon \to 0}  \Gamma_{\pm\epsilon} : = \Gamma_{\pm 0} 
\end{align}
exist  (note that $\Gamma_{\pm\epsilon} $ does not depend on  $g, r$ and $z$)  and they are given by
\begin{align}\label{LL2}
 \Gamma_{\pm 0}  =  \mp \pi i \theta(0) + \mathcal{P}  \int_{m- e_1}^\infty \theta(x)/x dx,
\end{align}
where, for $\tau >m-e_1$, we define  
\begin{align}\label{LL3}
\theta(\tau) : = 4 \pi    (e_1 +  \tau) 
 (    (e_1 +  \tau)^2 - m^2)^{1/2} f( (   (e_1 +  \tau)^2 - m^2)^{1/2}  )^2  .
\end{align}

   Note that $\theta(0) > 0$,  and hence, (see  \eqref{fmaszero}) 
\begin{align}
 \Im \Gamma_{\pm 0} = \mp  \pi \theta(0) \ne 0. 
\end{align} 
\end{lemma}
\begin{proof}
Note that $PVP=0$   and $PH_0P=e_1 P$.  We take $\epsilon\in (0,1)$ and $z \in \mathbb{R}$ with $ | z - e_1 | \leq r$.  
  We obtain from Definition \ref{def:feshbach}  that
\begin{align}
F_P(H-z\pm i \epsilon)=(e_1-z\pm i \epsilon)P -g^2   \hat \Gamma_{\pm\epsilon}  + R_\epsilon(g) ,
\end{align}
where
\begin{align}
\hat \Gamma_{\pm\epsilon} P:=PV\overline P (H_{0,\overline{P}}-e_1\pm i \epsilon)^{-1}\overline P VP
\end{align}
and 
\begin{align}
R_\epsilon(g)=g^2 PV\overline P \left((H_{0,\overline{P}}-e_1\pm i \epsilon)^{-1}  - (H_{\overline{P}}-z\pm i \epsilon)^{-1}\right)\overline P VP .
\end{align}
For $\kappa>0$ and sufficiently small $g, r>0$, Lemma \ref{lemma:mourrefeshbach} (iii) implies that $\norm{R_\epsilon(g)}\leq 
C  g^{2}(g^{1/2}  + r^{1/2})$.  We define $ \widetilde f_{\pm} (k) =   \frac{f(k)}{e_0+\omega(k)-e_1\pm i \epsilon} $ and calculate  
\begin{align}
\hat \Gamma_{\pm\epsilon} P = &PV\overline P (H_{0,\overline{P}}-e_1\pm i \epsilon)^{-1}\overline P \sigma_1\otimes a(f)^* P
=   PV\overline P (H_{0,\overline{P}}-e_1\pm i \epsilon)^{-1}  \varphi_0\otimes  f 
\notag \\ =& 
 PV\overline P   \varphi_0\otimes \widetilde f_{\pm}
=  \int \mathrm{d}^3k\,  \frac{f(k)^2}{\omega(k)-e_1\pm i \epsilon} P ,
\end{align}  
where we recall $e_0=0$.
This together with the definition of $\Gamma_{\pm\epsilon}$ in \eqref{eq:Gamma} completes the first part of the proof.

In the following we compute the limits as $\epsilon $ tends to zero of  $\Gamma_{\pm\epsilon}$. This is actually a consequence of the   Sokhotski-Plemelj theorem, we calculate using the changes of variables 
$s = (r^2 + m^2)^{1/2}  $  and $\tau = s - e_1 $  (we recall that we identify $f(k) \equiv f(|k|) $ and we do the same with $\omega$): 
\begin{align} \label{MM0}
  \int & \mathrm{d}^3k\,  \frac{f(k)^2}{\omega(k)-e_1\pm i \epsilon} = 
 4 \pi  \int_0^\infty \mathrm{d}r \,  r^2 f(r)^2 \frac{ 1 }{ \omega(r) - e_1 \pm  i  \epsilon } 
 \\ \notag = &      
  4 \pi   \int_m^\infty \mathrm{d}s \, s (s^2 - m^2)^{1/2} f( (s^2 - m^2)^{1/2}  )^2 \frac{  1  }{ (s - e_1)  \pm i \epsilon } 
  \\ \notag = &      
 4 \pi  \int_{m- e_1}^\infty \mathrm{d}\tau \,  (e_1 +  \tau) 
 (    (e_1 +  \tau)^2 - m^2)^{1/2} f( (   (e_1 +  \tau)^2 - m^2)^{1/2}  )^2 \frac{  1}{ \tau \pm  i \epsilon  }  .
\end{align}
Using  \eqref{MM0} and the  Sokhotski-Plemelj theorem,  we obtain that 
\begin{align}\label{MM2}
\lim_{\epsilon \to 0} \int & \mathrm{d}^3k\,  \frac{f(k)^2}{\omega(k)-e_1\pm i \epsilon} =
\mp \pi i \theta(0) + \mathcal{P}  \int_{m- e_1}^\infty \mathrm{d}x\, \theta(x)/x ,
\end{align}
and thereby, we complete the proof.
\end{proof}
\begin{corollary}
\label{Cor}
For sufficiently small $g, r >0$, small enough $\epsilon > 0$ (depending on $g$) and $z \in \mathbb{R}$ with $ | z - e_1 | \leq r$, the following holds true 
\begin{align}
P (H-z\pm i \epsilon)^{-1} P= (e_1-z-g^2 \Gamma_{\pm 0} )^{-1} P + \widetilde R(\epsilon, g, r),
\end{align}
where 
\begin{align}
\norm{ \widetilde  R(\epsilon, g, r)}    \leq   C  ( g^{1/2} + r^{1/2})  \Big |  \frac{1 }{  e_1-z-g^2 \Gamma_{\pm 0}   }   \Big |, 
\end{align}
 and $C$ does not depend on $ \epsilon, g , r $ and $z$; c.f. Remark \ref{rem:const}.
\end{corollary}
\begin{proof}
It follows from  \cite[Eq.\ (IV.13)]{bfs3}  that 
\begin{align}
P (H-z\pm i \epsilon)^{-1} P = F_P(H-z\pm i \epsilon)^{-1}, 
\end{align}
which is invertible for small enough $\epsilon $, $r$ and $g$ (this is a consequence of Lemma \ref{lemma:feshbachest}, we recall that $\Im  \Gamma_{\pm 0} \ne 0 $). 
We use Neumann series and Lemma \ref{lemma:feshbachest} to get
\begin{align} 
\big \| F_P(H-z\pm i \epsilon)^{-1} - &  (e_1-z-g^2 \Gamma_{\pm 0} )^{-1} P \big \| 
\\ \notag \leq &  \Big |  \frac{1 }{  e_1-z-g^2 \Gamma_{\pm 0}   }   \Big |
 \sum^\infty_{n=1}  \Big \|  
\frac{ R_\epsilon(g, r)  \pm i  \epsilon +   g^2  \Gamma_{\pm 0} - 
  g^2  \Gamma_{\pm \epsilon}   }{ e_1-z-g^2 \Gamma_{\pm 0}} \Big \|^{n} \\ \notag  
 \leq  & C  ( g^{1/2} + r^{1/2})  \Big |  \frac{1 }{  e_1-z-g^2 \Gamma_{\pm 0}   }   \Big |
   ,
\end{align}
for small enough $g, \epsilon$ and $r$ (we can take, for example, $\epsilon \leq g^{5/2}  $  and so small such that $| \Gamma_{\pm 0} - \Gamma_{\pm \epsilon}   | \leq  g^{1/2} $).  
\end{proof}
In addition, we present an easy formula for a certain matrix element of the time evolution restricted to a spectral subset. 
\begin{lemma}\label{PP}
We set $\Phi_1:=\varphi_1\otimes \Omega$. For every $s > 0$, we have
\begin{align}
\label{eq:stoneput}
\left\langle  \Phi_1, e^{-itH} \chi_s(H) \Phi_1\right\rangle 
&=   \pi^{-1}\lim\limits_{\epsilon\to 0^+}\int_\R \mathrm{d}r \, \chi_s(r) e^{-itr}  \Im  \left\langle  \Phi_1,  (H -  r - i\epsilon)^{-1}\Phi_1\right\rangle 
.
\end{align}   
\end{lemma}
\begin{proof}
The result follows from the spectral theorem and the next calculation
\begin{align}
e^{-i t \lambda } \chi_s(\lambda) = & \lim_{\epsilon \to 0}\frac{1}{\pi}  \int_{\mathbb{R}} dr
  e^{-i t (\lambda + \epsilon r)  } \chi_s(\lambda + \epsilon r) \frac{  1  }{r^2 + 1}
  = \lim_{\epsilon \to 0}\frac{1}{\pi}  \int_{\mathbb{R}}
 dr  e^{-i t r  }  \chi_s(r) \frac{  \epsilon  }{  ( r - \lambda     )^2 + \epsilon^2}\notag  \\ 
 = &  \lim_{\epsilon \to 0}\frac{1}{\pi}  \int_{\mathbb{R}}
 dr  e^{-i t r  } \chi_s(r) \Im  \frac{1}{ \lambda - r - i \epsilon} . 
\end{align}
\end{proof}
The following formula strongly relies on the previous results in this section and it is a crucial ingredient for the proof of the main theorem.
\begin{lemma}
\label{lemma:matrixelement}
  For sufficiently small $g>0$, $s$ as in Definition \ref{def.chi} and  Lemma \ref{lemma:specproj} sufficiently small,  and all $t\in\R$, the following holds true 
\begin{align}
\label{eq:stoneformula}
&\left\langle  \Phi_1, e^{-itH} \Phi_1 \right\rangle 
=\pi^{-1}\int_{\R} \mathrm{d}z \, e^{-itz}  \Im  (e_1-z-g^2 \Gamma_{- 0} )^{-1} + r_0 (g, s) 
,
\end{align}
where  
\begin{align}
|r_0(g, s)|\leq  C \left( (g^{1/2} + s^{1/2} )  |\log(g)| +   gs^{-1} \right), 
\end{align}
 and we recall $\Phi_1=\varphi_1\otimes \Omega$.  The constant $C$ does not depend on $g$, $s$ and $t$.
\end{lemma}
\begin{proof}
The spectral calculus implies  $\chi(H_0)\Phi_1=\Phi_1$, and hence, it follows from Lemma \ref{lemma:specproj} that 
\begin{align}
\label{proof:stoneformula}
&\left\langle  \Phi_1, e^{-itH} \Phi_1 \right\rangle=  \left\langle  \Phi_1, e^{-itH} \chi_s(H)\Phi_1 \right\rangle  + r_1(g, s),  \quad \text{where} \quad |r_1(g, s)|\leq C gs^{-1}  .
\end{align}
Lemma \ref{PP} yields
\begin{align}
\label{eq:stone}
\left\langle  \Phi_1, e^{-itH} \chi_s(H) \Phi_1\right\rangle 
&=  \pi^{-1}\lim\limits_{\epsilon\to 0^+}\int_\R \mathrm{d}z \, \chi_s(z) e^{-itz}  \Im  \left\langle  \Phi_1, P (H -  z  -   i\epsilon)^{-1} P \Phi_1\right\rangle 
.
\end{align}  
We calculate: 
\begin{align}
\label{eq:stone1}
\left\langle  \Phi_1, e^{-itH}\chi_s(H) \Phi_1 \right\rangle 
  = &   
 \pi^{-1}  \lim_{\epsilon\to 0^+}  \left( \int_{\R} \mathrm{d}z \,   e^{-itz}  \Im  (e_1-z-g^2 \Gamma_{- 0} )^{-1} + r_2(g,  \epsilon ,  s) + r_3(g, s) \right) , 
\end{align}
where  
\begin{align}
\label{r0}
r_2(g, \epsilon, s)=  \pi^{-1}  \int_{\R} \mathrm{d}z \, \chi_s (z) e^{-itz}  \Im \Big\langle \Phi_1,\Big (   P (H -  z  -   i\epsilon)^{-1} P   -  (e_1-z-g^2 \Gamma_{- 0} )^{-1}  \Big )\Phi_1\Big \rangle    
\end{align}
and
\begin{align}
\label{r00}
r_3(g, s)=  \pi^{-1}\int_{\R} \mathrm{d}z \, (1- \chi_s(z) )e^{-itz}  \Im   (e_1-z-g^2 \Gamma_{+ 0} )^{-1} .
\end{align}
Now, we use Corollary \ref{Cor}, for sufficiently small $s$, to get
\begin{align}
\Big | \chi_s (z) e^{-itz}  \Im \Big\langle \Phi_1,\Big (   P (H -  z  -   i\epsilon)^{-1} P   &-   (e_1-z-g^2 \Gamma_{- 0} )^{-1}  \Big )\Phi_1\Big \rangle    \Big |
 \\ \notag
&\leq C (g^{1/2} + s^{1/2} )  \left|   e_1-z-g^2 \Gamma_{- 0}  \right|^{-1}  \chi_s (z) .
\end{align}
This together with \eqref{r0}  and Definition \ref{def.chi}  yields then  that 
\begin{align}
\label{r'}
|r_2(g, \epsilon, s)|\leq & C (g^{1/2} + s^{1/2} )  \int \mathrm{d}z\,   \chi_s (z) 
\left|   e_1-z-g^2 \Gamma_{- 0}  \right|^{-1} 
\\ \notag   
=&   C (g^{1/2} + s^{1/2} ) 
\int \mathrm{d}z \,  \chi ((z-e_1)/s+e_1)  \Big ( (e_1 - z  - g^2 \Re \Gamma_{- 0})^2 +  g^4 (\Im \Gamma_{- 0})^2    \Big )^{-1/2} 
\\ \notag  
  \leq     &  C (g^{1/2} + s^{1/2} ) \int^{\frac{3}{4}\delta s-g^2 \Re \Gamma_{-0}}_{-\frac{3}{4}\delta s-g^2 \Re \Gamma_{-0}}  dr \frac{1}{g^2}\frac{1}{( (\frac{r}{g^2})^2  +  (\Im \Gamma_{- 0})^2)^{1/2} } 
\\ \notag   
\leq &  C (g^{1/2} + s^{1/2} ) \int_{|r|\leq c s g^{-2}}  dr \frac{1}{( r^2  +  (\Im \Gamma_{- 0})^2)^{1/2} } 
,
\end{align}
where the last step follows for $g>0$ sufficiently small and some constant $c>0$. Here, we recall from Definition \ref{def.chi} that  $g^2\leq s\leq g^{\kappa}$ for some $0<\kappa<2$. 
Employing that $2\sqrt{x^2+y^2}\geq |x|+|y|$, we find   a constant $C>0$ such that
\begin{align}
\label{r}
|r_2(g, \epsilon, s)|
\leq &  C (g^{1/2} + s^{1/2} )  |\log(g)|
.
\end{align}
Moreover, it follows from \eqref{r00} together with the definition of $\chi$ and $0\leq \chi \leq 1$ that
there is a  constant  $c > 0$ such that
\begin{align}
| r_3(g, s)|\leq & \pi^{-1}\int \mathrm{d}z\,  (1 - \chi_s (z)) \left|  \Im (  e_1-z-g^2 \Gamma_{-0} )^{-1} \right|
\\ \notag \leq &   \pi^{-1} g^2 \Im \Gamma_{+ 0}    \int (1 - \chi_s (z)) \frac{1}{  (e_1 - z  - g^2 \Re \Gamma_{- 0})^2 +  g^4 \Im \Gamma_{- 0}^2     } 
\\ \notag   \leq   &  C g^2  \int_{ | r  | \geq  c s}  dr \frac{1}{g^4}\frac{1}{ (\frac{r}{g^2})^2  +  \Im \Gamma_{- 0}^2  }   
  =     C  \int_{ | x  | \geq  c s/g^2 }  dx \frac{1}{  x^2  +  \Im \Gamma_{- 0}^2  }    
\leq C g^2s^{-1}  .
\end{align} 
\end{proof}

\section{Proof of the main result}
\label{sec:proof-mainresult}

In this section we provide a proof of the   main result; c.f.\ Theorem~\ref{FKcor}. 
\begin{proof}[Proof of Theorem \ref{FKcor}]   We start with  \eqref{W} and use \eqref{amenospsi}:  
\begin{align}
T(h, l)
&=
-2\pi i g\norm{\Psi_{\lambda_0}}^{-2} \left\langle a_-(W) \sigma_1 \Psi_{\lambda_0}, \Psi_{\lambda_0}\right\rangle
=-2\pi i g\norm{\Psi_{\lambda_0}}^{-2} \left\langle [a_-(W) ,\sigma_1] \Psi_{\lambda_0}, \Psi_{\lambda_0}\right\rangle ,
\end{align}
  It follows from \eqref{amenosintegral} that  
\begin{align}
\label{eq:thl0tt}
T(h,l)
&=
 2\pi(ig)^2\norm{\Psi_{\lambda_0}}^{-2} \int_{-\infty}^0 \mathrm{d}t   \,
\overline{\langle W_t,f\rangle_2}
\left\langle \left[e^{itH}\sigma_1
e^{-itH},  \sigma_1\right]  \Psi_{\lambda_0}, \Psi_{\lambda_0}\right\rangle 
\notag \\
&=   2\pi  g^2\norm{\Psi_{\lambda_0}}^{-2} \int^{\infty}_0 \mathrm{d}t \,
    \langle f, W_{-t}\rangle_2
    \left\langle \left[e^{-itH}\sigma_1
e^{itH},  \sigma_1\right]  \Psi_{\lambda_0}, \Psi_{\lambda_0}\right\rangle 
\notag \\
&= 2 \pi   \norm{\Psi_{\lambda_0}}^{-2} \left(  T^{(1)}- T^{(2)} \right) 
,
\end{align}
where  we recall 	the notation $W_s(k)=e^{-is\omega(k)}W(k)$ and  use the abbreviations 
\begin{align}
\label{eq:first-termtt}
T^{(1)} :&= g^2 \int_0^\infty  \mathrm{d}t  \int \mathrm{d^3} k   \,
W(k)f(k) e^{it(\omega(k)+\lambda_0)}  \left\langle   \sigma_1 \Psi_{\lambda_0},
e^{-itH} \sigma_1\Psi_{\lambda_0}\right\rangle 
\\ \notag 
&=  g^2\int_0^\infty  \mathrm{d}t  \int_0^\infty \mathrm{d} r \, G(r) e^{it(\omega(r)+\lambda_0)}  \left\langle   \sigma_1 \Psi_{\lambda_0},
e^{-itH} \sigma_1\Psi_{\lambda_0}\right\rangle 
\\ \notag 
&   =  g^2\int_0^\infty  \mathrm{d}t  \, \zeta(t)  \left\langle   \sigma_1 \Psi_{\lambda_0},
e^{-itH} \sigma_1\Psi_{\lambda_0}\right\rangle 
\end{align}
and 
\begin{align}
T^{(2)} := g^2  \int_0^\infty  \mathrm{d}t  \int_0^\infty \mathrm{d} r\,   G(r) e^{it(\omega(r)-\lambda_0)}     \left\langle \sigma_1  \Psi_{\lambda_0},
e^{itH}\sigma_1\Psi_{\lambda_0}\right\rangle .
\label{eq:second-termtt}
\end{align}
Here, we  changed to spherical coordinates $k=(r,\Sigma)$ and  take:
\begin{align}
\label{def:W1sttt}
 G(r)=   \int \mathrm{d}\Sigma \mathrm{d}\Sigma' \,  r^4  \overline{h(r,\Sigma)} l(r,\Sigma') f(r)^2, 
\qquad  \zeta (t): =   \int_0^\infty \mathrm{d} r \, G(r) e^{it(\omega(r)+\lambda_0)} . 
\end{align}
Moreover, we observe that $G\in\mathit C_c^\infty(\R\setminus \{0\}, \C)$. Notice that  an integration by parts   (using that $e^{i \theta \omega(r)} =   \frac{\partial}{\partial r} \Big ( e^{i \theta \omega(r)}  \frac{1}{i \theta \frac{\partial}{\partial \theta} \omega(r) } \Big ) 
   -    e^{i \theta \omega(r)} \frac{\partial}{\partial r}   \Big ( \frac{1}{i \theta \frac{\partial}{\partial \theta} \omega(r)} \Big ) $)    ensures  that  
\begin{align} 
\label{intparts2}
\left| \zeta (t)  \right|\leq C/(1+t^2), \qquad  \forall t \in \R,
\end{align} 
which guarantees the existence of the integrals in \eqref{eq:first-termtt} and \eqref{eq:second-termtt}.

 Recall $\Phi_1=\varphi_1\otimes \Omega$ (see Lemma \ref{PP}). 
It follows from   Proposition  \ref{prop:gs} that
\begin{align}
\label{approxvec}
& \left\langle \sigma_1  \Psi_{\lambda_0},
e^{-isH}\sigma_1\Psi_{\lambda_0}\right\rangle= \left\langle \Phi_1,
e^{-isH}\Phi_1\right\rangle +   \rho_1(g) ,
\end{align}
 where   $|\rho_1(g)|\leq C  g$. 
Moreover, we recall that Lemma \ref{lemma:matrixelement} states  that  
\begin{align}
\label{eq.matel}
&\left\langle  \Phi_1, e^{-itH} \Phi_1 \right\rangle 
=\pi^{-1}\int_{\R} \mathrm{d}z \, e^{-itz}  \Im  (e_1-z-g^2 \Gamma_{- 0} )^{-1} + r_0 (g, s) 
,
\end{align}
where  
\begin{align}
|r_0(g, s)|\leq C \Big ( \left(g^{1/2} + s^{1/2} \right) |\log(g)|  +   gs^{-1} \Big ).  
\end{align}  
Note that \cite[Remark 4.8]{bdh-scat} implies that the first term in \eqref{eq.matel} is bounded by a constant as $g\to0^+$  (this actually follows from computing the integral). 
Then, \eqref{eq:first-termtt} together with \eqref{approxvec} and \eqref{eq.matel} yields
\begin{align}
\label{49}
T^{(1)} = T^{(1)}_0 +  R_1(g, s)  ,
\end{align}
where
\begin{align}
\label{t10}
  T^{(1)}_0 :=& \pi^{-1}g^2\int_0^\infty  \mathrm{d}t\,      \zeta(t) 
  \int_{\R} \mathrm{d}z \, e^{-itz}   \Im  (e_1-z-g^2 \Gamma_{- 0} )^{-1} ,
\end{align}
and  $|R_1(g, s)|\leq C
 g^2 (  (g^{1/2} + s^{1/2} ) |\log(g)|  +  gs^{-1}  )$  for some constant $C>0$.
 As  $ \Im   \Gamma_{- 0} > 0   $ and $  \Im  (e_1-z-g^2 \Gamma_{- 0} )^{-1} $ decays as
$ |z|^{-2}$ at infinity,  
 \eqref{t10}  is absolutely integrable, and consequently,  Fubini's theorem allows for interchanging the order of integration. Similarly, we argue that we can apply the  dominated convergence theorem and conclude 
\begin{align}
\label{eq:lim101}
  T^{(1)}_0 =&  \lim\limits_{\eta\to 0^+}T^{(1)}_0(\eta)  , 
\end{align}
where
\begin{align}
  T^{(1)}_0(\eta) =& \pi^{-1}g^2 \int_\R \mathrm{d}z \,     \Im  (e_1-z-g^2 \Gamma_{- 0} )^{-1}   \int_0^\infty  \mathrm{d}t  \int_0^\infty \mathrm{d} r \, G(r) e^{it(\omega(r)+\lambda_0-z+i\eta)} 
  \\ \notag = & \pi^{-1}g^2 \int_\R \mathrm{d}z \,   \Im  (e_1-z-g^2 \Gamma_{- 0} )^{-1}   \int_0^\infty  \mathrm{d}t\,    \zeta(t)  e^{-t \eta}  e^{-it z}.   
\end{align}
Again, Fubini's theorem yields for $Q>0$
\begin{align}
\label{eq:Q}
\int_0^Q  \mathrm{d}t  \int_0^\infty \mathrm{d} r \, G(r) e^{it(\omega(r)+\lambda_0-z+i\eta)} 
= i\int_0^\infty \mathrm{d} r \, \frac{G(r)}{\omega(r)+\lambda_0-z+i\eta} \left( 1-  e^{iQ(\omega(r)+\lambda_0-z+i\eta)}  \right).
\end{align}
Moreover, for all $\eta>0$, we obtain  by the integration by parts formula  (see above \eqref{intparts2}) 
together with $G\in\mathit C_c^\infty(\R\setminus \{0\}, \C)$ that  there is a constant  $C(\eta, g)>0$  such that
\begin{align}
\left| \int_0^\infty \mathrm{d} r \, \frac{G(r)}{\omega(r)+\lambda_0-z+i\eta} e^{iQ(\omega(r)+\lambda_0-z+i\eta)}  \right| \leq   C(\eta, g)Q^{-1} ,
\end{align}
and consequently,  \eqref{eq:Q} implies that
\begin{align}
\int_0^\infty  \mathrm{d}s  \int_0^\infty \mathrm{d} r \, G(r) e^{is(\omega(r)+\lambda_0-z+i\eta)} 
&= i\int_0^\infty \mathrm{d} r \, \frac{G(r)}{\omega(r)+\lambda_0-z+i\eta}.
\end{align}
This together with Fubini's theorem yields that
\begin{align}
\label{eq:lim102}
 & T^{(1)}_0(\eta) = i \pi^{-1} g^2 \int_0^\infty \mathrm{d} r \,   G(r)  \int_\R \mathrm{d}z \, \Im  (e_1-z-g^2 \Gamma_{- 0} )^{-1}    \frac{1}{\omega(r)+\lambda_0-z+i\eta}  .
\end{align}
For $a>e_1$, we define $\mathcal Q_a:=[-a,a]  \cup   \{ae^{-i\varphi}:\varphi\in [0,\pi]\}\subset \overline{\C^-}$  to be a closed contour with mathematical negative orientation.   Note that, for real z, as in  \eqref{eq:lim102},    
\begin{align}
  \Im  (e_1-z-g^2 \Gamma_{- 0} )^{-1} =   \frac{1}{2i}  \Big ( (e_1-z-g^2 \Gamma_{- 0} )^{-1} - 
  (e_1-z-g^2 \overline{ \Gamma_{- 0}} )^{-1}
  \Big ), 
\end{align}
i.e. we do not  conjugate $ z $. We extend the formula above, in a meromorphic way, to the lower
half of the complex plane. We obtain, for small enough $\eta$, using the residue theorem that   
\begin{align}
& \int_{\R} \mathrm{d}z \,   \Im  (e_1-z-g^2 \Gamma_{- 0} )^{-1}  \frac{1}{\omega(r)+\lambda_0-z+i\eta} 
 \notag \\
 &=(2i)^{-1} \lim\limits_{a\to\infty} \int_{\mathcal Q_a} \mathrm{d}z \,  \frac{1}{\omega(r)+\lambda_0-z+i\eta}  \Big ( (e_1-z-g^2 \Gamma_{- 0} )^{-1} - 
  (e_1-z-g^2 \overline{ \Gamma_{- 0}} )^{-1}
  \Big ) 
 \notag \\ 
 &=  \frac{\pi}{\omega(r)+  \lambda_0- \big (  e_1 -g^2 \Gamma_{- 0}  \big )   +i\eta} . 
\end{align}
This together with \eqref{eq:lim102} yields that 
\begin{align}
 \lim\limits_{\eta\to 0^+} T^{(1)}_0(\eta) &=  \lim\limits_{\eta\to 0^+} \int_0^\infty \mathrm{d} r \, \frac{ i  g^2   G(r)}{    \omega(r)+  \lambda_0- \big (  e_1 -g^2 \Gamma_{- 0}  \big )   +i\eta    }
 \\ \notag 
 &= \int_0^\infty \mathrm{d} r \, \frac{ i  g^2  G(r)}{\omega(r)+  \lambda_0- \big (  e_1 -g^2 \Gamma_{- 0}  \big )   },
\end{align} 
where in the last step we applied  
the dominated convergence theorem which is justified because $G\in \mathit C_c^\infty(\R\setminus\{0\}, \C)$.
Consequently, it follows from \eqref{49} and \eqref{eq:lim101}  that
\begin{align}
\label{t1}
  T^{(1)} =&  
     i g^2 \int_0^\infty \mathrm{d} r \,     \frac{  G(r)}{\omega(r)+  \lambda_0- \big (  e_1 -g^2 \Gamma_{- 0}  \big )   }  +R_1(g, s)  ,
\end{align}
where we recall that  $|R_1(g, s)|\leq C
 g^2 ( (s^{1/2} + g^{1/2}) |\log(g)|  + gs^{-1} )$.  
Analogously, we obtain  
\begin{align}
\label{t2}
  T^{(2)} =& 
  i g^2 \int \mathrm{d} r \,       \frac{  G(r)}{\omega(r)-  \lambda_0+  \big (  e_1 -g^2 \overline{ \Gamma_{- 0}}  \big )   } +R_2(g, s)  ,
\end{align}
and $|R_2(g, s)|\leq C
 g^2 ( (s^{1/2} + g^{1/2}) |\log(g)|  + gs^{-1} )$ for some constant $C$.  Finally, we conclude from \eqref{t1} and  \eqref{t2} together with \eqref{eq:thl0tt} that 
\begin{align}
T(h,l)
= &  2 \pi i g^2 \norm{\Psi_{\lambda_0}}^{-2}  \int \mathrm{d} r \, 
   \left( \frac{  G(r)}{\omega(r)+  \lambda_0- \big (  e_1 -g^2 \Gamma_{- 0}  \big )   } - \frac{  G(r)}{\omega(r)-  \lambda_0+  \big (  e_1 -g^2 \overline{ \Gamma_{- 0}}  \big )   }\right) \notag \\ \notag &  +R(g, s)
\notag \\ 
= & 4 \pi i g^2 \norm{\Psi_{\lambda_0}}^{-2}  \int \mathrm{d} r \,    
\frac{G(r) \left( e_1-g^2\Re  \Gamma_{+ 0}     -\lambda_0 \right)}{
\left( \omega(r)+  \lambda_0- \big (  e_1 -g^2 \Gamma_{- 0}  \big ) \right)\left(\omega(r)-  \lambda_0+  \big (  e_1 -g^2 \overline{ \Gamma_{- 0}}  \big ) \right)}   \notag \\ & +R(g, s)
,
\end{align}   
where  $ R(g, s):=R_1(g, s)+R_2(g, s)$.   Hence, there is a  constant $C>0$ such that
 $|R(g, s)|\leq C
 g^2 (  (s^{1/2} + g^{1/2})  |\log(g)|   +  gs^{-1}  )$. We take $s = g^{2/3}$  and obtain that 
 $   |R(g, s)|\leq C g^2  g^{1/3}    |\log(g)|     $.
     This completes the proof.
\end{proof}

\section{Mourre Theory and the Limiting Absorption Principle}
\label{app:limab}

In this section we present a proof of  Proposition \ref{prop:absenceofacspecclose} (ii) and (iii).   
Although Mourre theory is a standard tool to prove limiting absorption principles, in this section we do not address the usual procedures because we prove perturbative  results in the spirit of \cite{ahs,fms} (see Proposition \ref{prop:absenceofacspecclose}  (iii)).  Note that in \cite{fms} an abstract family of Hamiltonians is studied.  

The main result of this section is  Proposition \ref{prop:absenceofacspecclose}  (iii). Despite the fact that  Proposition \ref{prop:absenceofacspecclose}  (ii) is standard, we also prove it because we need  it to prove Proposition \ref{prop:absenceofacspecclose}  (iii). Some other well-known estimates in the context of Mourre theory are not proven in this section -- we will give instead proper references. 

We also mention that we do not employ the original techniques of Mourre to study domain problems and commutators (see \cite{mourre,cycon}). Instead, we directly dilate the operators at stake: our approach is close to the usual one based on the theory of operators of class $C^k$ with respect to a self-adjoint conjugate operator (see \cite{bookabg,amrein}), but, in our paper, given the explicit form of the operators at stake, we do not need to rely on this theory and we give a more transparent presentation.                          

In this section we address the limiting absorption principle, i.e. we study the behavior of the resolvent operator $ (H_{\overline P}-(z\pm i \epsilon))^{-1}   $ as $\epsilon  >0 $ tends to $0$ and $z$ belongs to the interval 
\begin{align}\label{DefI}
I: =[e_1-\delta /4, e_1+\delta /4].
\end{align}
Of course, the norm of   $ (H_{\overline P}-(z\pm i \epsilon))^{-1}   $ tends to infinity as $\epsilon$ tends to zero.  Then, controlling its behavior requires   restricting its domain, and this is achieved by multiplying by the operator
\begin{align}\label{Defirho}
  \rho:=  \left\langle     \mathrm{d}\Gamma(D)  \right\rangle^{-1}.
\end{align}
Our goal is to obtain uniform norm-bounds for $ \rho  (H_{\overline P}-(z\pm i \epsilon) )^{-1}  \rho  $ and regularity properties with respect to $g $  (this is what we call above perturbative Mourre theory)  and $z$.  

Intuitively, one might consider the operator  $  H_{\overline P}-z  $ as a real quantity because it is self adjoint.  One of the clever ideas of Mourre is to add to $  H_{\overline P}-(z  \pm i \epsilon ) $  a non-zero  imaginary part of size $\eta >  0$ and  sign $\pm$  (according to $ \pm i \epsilon  $). Then, the resulting operator ( $  H^{\pm \eta}_{\overline P}-z_{\pm \epsilon }   $ -- see  \eqref{DefiHoP} below) can be intuitively regarded as a real quantity plus 
$\mp i (\epsilon + \eta)$. It is, therefore, invertible and the norm of its inverse is uniformly bounded with respect to $\epsilon$. Our goal is to study the behavior of the resolvent operator associated to $H^{\pm \eta}_{\overline P}-z_{\pm \epsilon } $ as $\epsilon $ and $\eta$ tend to zero. More precisely, the imaginary part that we refer above is given by the operator $   \mp i \eta M^2 $, where $\eta$ is a strictly positive small enough real number  and (see Lemma \ref{lemma:comestclose})
\begin{align}\label{MOURRE}
 M^2:=\chi(H_{\overline P}) [H_{\overline P},  i    \mathrm{d}\Gamma(D)]^0\chi(H_{\overline P}) \geq \alpha \chi(H_{\overline{P}})^{2},  
\end{align}  
which is a bounded operator (see Remark  \ref{DEfiCommut}).  
We properly select $\rho$ as a function of $  \mathrm{d}\Gamma(D) $ because 
$  \rho  \mathrm{d}\Gamma(D) $ is bounded. This allows us to control the unbounded operator $  \mathrm{d}\Gamma(D) $ in the above commutator. The other operator in this commutator is chosen in order to cancel resolvents (see  \eqref{COMM1} and \eqref{COMM2} below for the limiting absorption principle, and   \eqref{As} for perturbative results). 

We define the operators (for $z  \in I$)
\begin{align}\label{DefiHoP}
H^{\pm \eta }_{\overline{P} } := H_{\overline{P}}  \mp  i \eta M^2,   \qquad R^{\pm\eta}(z_{\pm\epsilon}) = \left( H^{\pm\eta}_{\overline P}-z_{\pm\epsilon}  \right)^{-1}, \quad  z_{\pm \epsilon} :=  z \pm i \epsilon . 
\end{align}
It is a standard result that $ H^{\pm\eta}_{\overline P}-z_{\pm\epsilon}   $ is invertible (with bounded inverse) -- see \cite{cycon} -- and that $R^{\pm\eta}(z_{\pm\epsilon})$ is continuous at $\eta = 0$ and derivable with respect to $\eta$, for $\eta > 0 $ small enough.  Its derivative is given by
\begin{align}
\label{Teo5i}
\mathrm{d}/\mathrm{d}\eta \, R^{\pm\eta}(z_{\pm\epsilon})    =\pm iR^{\pm\eta}(z_{\pm\epsilon})  M^2     R^{\pm \eta}(z_{\pm\epsilon}) , \qquad \forall \eta\in (0,\boldsymbol \eta).
\end{align}
For the convenience of the reader we give a proof of this in Appendix \ref{Standard}  below (see also \cite{cycon}). Moreover, if we multiply 
 $ R^{\pm\eta}(z_{\pm\epsilon}) $ by an operator that localizes the spectral region of $ H_{\overline{P}} $ far away from $z$, we get a bounded operator which satisfies:
\begin{align}\label{chibarR}
\norm{(H_{\overline P}+i)\overline \chi(H_{\overline P})R^{\pm\eta}(z_{\pm\epsilon})}\leq C,
\end{align}
where $\overline \chi=1-\chi$. This is proven in Appendix \ref{Standard} (see also \cite{cycon}).

As announced above, it follows that the norm of $ R^{\pm\eta}(z_{\pm\epsilon})   $ can be uniformly bounded (with respect to $\epsilon$).  Actually, the following estimate holds:
\begin{align}\label{TEo5iV}
\| (H_{\overline{P}} + i )  R^{\pm\eta}(z_{\pm\epsilon})  \| \leq C/\eta, 
\end{align}   
where $C$ does not depend on $z, \epsilon $ and $g$ (see \cite{cycon} and Section \ref{Standard}). 

Estimate  \eqref{TEo5iV} itself is not enough because we still have the singularity 
$C/\eta$ and we need to consider the operator $\rho$, otherwise we cannot expect to have a limiting absorption principle - this is explained above. 
For this reason, we define
\begin{align}\label{DefiF}
F^{\pm\eta}(z_{\pm\epsilon}):= \rho   R^{\pm\eta}(z_{\pm\epsilon})    \rho
\end{align} 
and  get a better estimate which is a key ingredient of Mourre theory. Note that this is the only place where the Mourre estimate (see  \eqref{MOURRE}) is used: 
\begin{align}\label{Teo5v}
\norm{  (H_{\overline P}+i)   R^{\pm\eta}(z_{\pm\epsilon})\rho}\leq C\left(1+\eta^{-1/2}\norm{ F^{\pm\eta}(z_{\pm\epsilon})}^{1/2}\right).
\end{align}
Eq.\  \eqref{Teo5v}  is a standard result (see, e.g., \cite{cycon}), but we prove it in Appendix \ref{Standard}.  Looking at  \eqref{TEo5iV} and \eqref{Teo5v}, it seams that we get again the unsatisfactory bound 
\begin{align}\label{TEoLimitingi}
\| F^{\pm\eta}(z_{\pm\epsilon}) \| \leq C/\eta. 
\end{align}
At this point, the line of reasoning becomes more subtle.  Actually, in the lines above we never use that $M^2$ is defined in terms of the commutator $   [H_{\overline P},  i    \mathrm{d}\Gamma(D)]^0 $. The only thing we utilize about 
 $M^2$ is that it satisfies the Mourre estimate \eqref{MOURRE}. All the material presented above in this section is standard and it can be directly deduced from the proofs in \cite{cycon}. Therefore, we do not include proofs of this in the present section. For the convenience of the reader we provide proofs in Appendix \ref{Standard}.

In this section we use all estimates and statements presented above (without proofs) and provide a detailed proof of  the limiting absorption principle (Proposition \ref{prop:absenceofacspecclose}-(ii)) and its perturbative version (Proposition \ref{prop:absenceofacspecclose}-(iii)).  The idea of the proof of Proposition \ref{prop:absenceofacspecclose}-(ii) (which  amounts to bound $\| F^{\pm\eta}(z_{\pm\epsilon}) \|$ by a constant) is quite simple, we just write 
$ F^{\pm\eta}(z_{\pm\epsilon}) $ as the integral of its derivative. Then, the difficult part is to estimate the referred derivative (Lemma \ref{keyest} below). This derivative consists of a sum of several terms and each of them is separately estimated. The most singular term is $Q_{1,1}$ defined in    \eqref{COMM1} below. The analysis of $Q_{1,1}$ is the  only part of the proof of Proposition \ref{prop:absenceofacspecclose}-(ii)  that requires that $M^2$ is defined in terms of the commutator $   [H_{\overline P},  i    \mathrm{d}\Gamma(D)]^0 $: 
we control the unbounded operator   $  \mathrm{d}\Gamma(D) $ using that  
$  \rho  \mathrm{d}\Gamma(D) $ is bounded and $ H_{\overline P} $ is important to cancel resolvent operators (see   \eqref{COMM2} below). 

As we mention above, the main result of this section is Proposition \ref{prop:absenceofacspecclose}-(iii). The proof of it follows the same strategy of the proof of item (ii), but it is substantially more complicated. Again, we study the terms we are interested in using that they are integrals of their derivatives.  The difficult part is to estimate the derivatives, which consist on several terms that must be analyzed separately. This is achieved in Lemma 
\ref{keyest1} below. 

Before we start with the proofs, we state two last results that we use in this section and prove in Appendix \ref{Nonstandard}: the operator
$  R^{\pm\eta}(z_{\pm\epsilon})  $ leaves the domain of   $  \mathrm{d}\Gamma(D) $ invariant. Moreover,  there is a bounded operator that we denote by 
\begin{align}\label{ComutGammaM2}
[ \mathrm{d}\Gamma(D),M^2]^0
\end{align}
that represents the quadratic form $  [ \mathrm{d}\Gamma(D),M^2] $. These results can be proved as in \cite{mourre,cycon} (defining a scale of Hilbert spaces and regularizing the generator of dilations) or \cite{bookabg,amrein} (using that the Hamiltonian is of class $C^k$ with respect to the generator of dilations). We provide a more direct proof in Appendix \ref{Nonstandard}.
\begin{remark}\label{zero}
The definitions and estimates introduced above in this section are also valid for the case $g = 0$. We distinguish this case by adding everywhere in our notations a subscript $0$. For example: 
$$M^2_0 : =  M^2|_{g= 0},  \hspace{1cm} 
 H_{0, \overline{P}} := H_{ \overline{P}}|_{g= 0}. $$        
\end{remark}
\begin{lemma}
\label{keyest}
For $g \geq 0, \boldsymbol \eta >0$ sufficiently small, 
 $\eta\in (0,\boldsymbol \eta)$,  $\epsilon\in (0,1)$, $z, z'\in I$ and $z_{\pm\epsilon}:=z\pm i\epsilon$, 
 \begin{align}
\norm{\mathrm{d}/\mathrm{d}\eta \, F^{\pm\eta}(z_{\pm\epsilon})}\leq C\left(   1+\eta^{-1/2}\norm{ F^{\pm\eta}(z_{\pm\epsilon})}^{1/2}+\norm{ F^{\pm\eta}(z_{\pm\epsilon})} \right)
.
\end{align}
\end{lemma}
\begin{proof}
 It follows from \eqref{MOURRE}, \eqref{Teo5i} and \eqref{DefiF} that 
\begin{align}
\label{qsplit}
\pm i\mathrm{d}/\mathrm{d}\eta \, F^{\pm\eta}(z_{\pm\epsilon})
= Q_1+Q_2+Q_3+Q_4,
\end{align}
where 
\begin{align}
Q_1:&=  -  \rho  R^{\pm\eta}(z_{\pm\epsilon})   [H_{\overline P},i  \mathrm{d}\Gamma(D)]^0      R^{\pm\eta}(z_{\pm\epsilon}) \rho
\\
Q_2:&=- \rho  R^{\pm\eta}(z_{\pm\epsilon}) \overline \chi(H_{\overline P}) [H_{\overline P},i  \mathrm{d}\Gamma(D)]^0 \overline \chi(H_{\overline P})     R^{\pm\eta}(z_{\pm\epsilon}) \rho
\\
Q_3:&= \rho  R^{\pm\eta}(z_{\pm\epsilon}) \overline \chi(H_{\overline P})  [H_{\overline P},i  \mathrm{d}\Gamma(D)]^0    R^{\pm\eta}(z_{\pm\epsilon}) \rho 
\\
Q_4:&= \rho  R^{\pm\eta}(z_{\pm\epsilon})  [H_{\overline P},i  \mathrm{d}\Gamma(D)]^0 \overline \chi(H_{\overline P})    R^{\pm\eta}(z_{\pm\epsilon}) \rho .
\end{align}
Remark \ref{DEfiCommut} and \eqref{chibarR} imply that     
\begin{align}
\label{crucial1234}
\norm{ [H_{\overline P},i \mathrm{d}\Gamma(D)]^0 \overline \chi(H_{\overline P} )R^{\pm\eta}(z_{\pm\epsilon})}\leq C.
\end{align}
This  yields that 
\begin{align}
\label{q2}
\norm{Q_2}
\leq C \norm{\rho R^{\pm\eta}(z_{\pm\epsilon}) \overline \chi(H_{\overline P})}
  \leq C ,
\end{align}
 where we use again \eqref{chibarR} (taking the adjoint).   
  Taking the adjoint in  \eqref{crucial1234}, it follows that  
\begin{align}
\label{q3}
\norm{Q_3}&
\leq  C\norm{ R^{\pm\eta}(z_{\pm\epsilon})\rho } \leq  C\left(1+\eta^{-1/2}\norm{ F^{\pm\eta}(z_{\pm\epsilon})}^{1/2}\right), 
\end{align}
 where we use \eqref{Teo5v}. 
  Similarly, taking the adjoint in  \eqref{Teo5v} we obtain that   
\begin{align}
\label{q4}
\norm{Q_4}&
\leq  C\left(1+\eta^{-1/2}\norm{ F^{\pm\eta}(z_{\pm\epsilon})}^{1/2}\right) 
  .
\end{align}
In the remainder of the proof, we estimate    $Q_1$.
 For $\phi, \psi \in \mathcal D( \mathrm{d}\Gamma(D))\cap \mathcal D(H_{\overline P})$,   Remark \ref{DEfiCommut} and  the fact that 
 $  R^{\pm\eta}(z_{\pm\epsilon})  $ leaves the domain of   $  \mathrm{d}\Gamma(D) $ invariant  (see above \eqref{ComutGammaM2})   allows us to write
\begin{align}
\label{q1split}
 \left\langle \phi, Q_{1} \psi \right\rangle =\left\langle \phi,  Q_{11}\psi \right\rangle  + \left\langle \phi, Q_{12}\psi \right\rangle  ,
\end{align}
where
\begin{align}\label{COMM1}
 \left\langle \phi, Q_{11}    \psi \right\rangle  :&= \left\langle \left( H_{\overline P}\pm i\eta M^2-z_{\mp \epsilon}  \right)  R^{\mp\eta}(z_{\mp\epsilon}) \rho  \phi,    i  \mathrm{d}\Gamma(D)      R^{\pm\eta}(z_{\pm\epsilon}) \rho\psi \right\rangle   
 \notag \\
 &-  \left\langle   \left(    - i  \mathrm{d}\Gamma(D) \right) R^{\mp\eta}(z_{\mp\epsilon}) \rho  \phi,    \left( H_{\overline P}\mp i\eta M^2-z_{\pm \epsilon}  \right)      R^{\pm\eta}(z_{\pm\epsilon}) \rho\psi \right\rangle    ,
\\
 \left\langle \phi, Q_{12} \psi \right\rangle  :&=\pm i\eta \bigg(  \left\langle  M^2   R^{\mp\eta}(z_{\mp\epsilon})  \rho \phi,     i    \mathrm{d}\Gamma(D)      R^{\pm\eta}(z_{\pm\epsilon}) \rho\psi \right\rangle
 \notag \\
 &-\left\langle  \left(-  i    \mathrm{d}\Gamma(D)   \right) R^{\mp\eta}(z_{\mp\epsilon})  \rho \phi,     M^2      R^{\pm\eta}(z_{\pm\epsilon}) \rho\psi \right\rangle   \bigg)
  .
\end{align}
Employing   that  $\norm{   \mathrm{d}\Gamma(D) \rho}\leq 1$, we find 
\begin{align}\label{COMM2}
\left| \left\langle \phi,  Q_{11} \psi \right\rangle \right| 
& = \left|    
\left\langle  \left(    - i  \mathrm{d}\Gamma(D) \right)  \rho  \phi,         R^{\pm\eta}(z_{\pm\epsilon}) \rho\psi \right\rangle   
 -  \left\langle  R^{\mp\eta}(z_{\mp\epsilon}) \rho  \phi,    i  \mathrm{d}\Gamma(D)  \rho\psi \right\rangle  
 \right|
\notag \\
 &\leq  \norm{\phi} \norm{\psi}\left( \norm{R^{\mp\eta}(z_{\mp\epsilon}) \rho}+   \norm{ R^{\pm\eta}(z_{\pm\epsilon}) \rho} \right).
\end{align}
It follows again from   \eqref{Teo5v}    that 
\begin{align}
\label{q11}
\left| \left\langle \phi,  Q_{11} \psi \right\rangle \right| &
\leq  C \norm{\phi} \norm{\psi} \left(1+\eta^{-1/2}\norm{ F^{\pm\eta}(z_{\pm\epsilon})}^{1/2}\right) 
  .
\end{align}
Furthermore, we estimate  (using again that  $  R^{\pm\eta}(z_{\pm\epsilon})  $ leaves the domain of   $  \mathrm{d}\Gamma(D) $ invariant and the text around \eqref{ComutGammaM2})  
\begin{align}\label{q12}
\left| \left\langle \phi,  Q_{12} \psi \right\rangle \right| 
&\leq \eta   \norm{\phi} \norm{\psi} \norm{ R^{\mp\eta}(z_{\mp\epsilon}) \rho}\norm{R^{\pm\eta}(z_{\pm\epsilon}) \rho} \norm{ [M^2,  \mathrm{d}\Gamma(D)]^0  }
 \\ &   \leq  
C \eta  \norm{\phi} \norm{\psi}  \left(1+\eta^{-1/2}\norm{ F^{\pm\eta}(z_{\pm\epsilon})}^{1/2}\right)^2
\notag \\
&\leq   C\norm{\phi} \norm{\psi}  \left(1+\norm{ F^{\pm\eta}(z_{\pm\epsilon})}\right)
  , 
\notag 
\end{align}
 where we use \eqref{Teo5v}. 
It follows  from \eqref{q11} together with \eqref{q12}, \eqref{q1split} and the density of $\mathcal D( \mathrm{d}\Gamma(D))\cap \mathcal D(H_{\overline P})$ in $\mathcal H$  that 
\begin{align}
\label{q1}
\norm{Q_1}\leq  C\left(1+\eta^{-1/2}\norm{ F^{\pm\eta}(z_{\pm\epsilon})}^{1/2}+\norm{ F^{\pm\eta}(z_{\pm\epsilon})}\right)  .
\end{align}
This together with \eqref{qsplit}, \eqref{q2}, \eqref{q3} and \eqref{q4} completes the proof. 
\end{proof}
\begin{proof}[Proof of Proposition \ref{prop:absenceofacspecclose} (ii)]
 Let $\eta\in (0,\boldsymbol \eta)$ (and $ \boldsymbol \eta  $ is sufficiently small). We use the fundamental theorem of calculus  
\begin{align}
 F^{\pm\eta}(z_{\pm\epsilon}) =F^{\pm \boldsymbol \eta}(z_{\pm\epsilon})  + \int^{\pm \eta}_{\pm \boldsymbol \eta}\mathrm{d}\tilde \eta \, \mathrm{d}/\mathrm{d}\tilde \eta \, F^{\pm \tilde \eta}(z_{\pm\epsilon}),
\end{align}
and \eqref{TEoLimitingi}  to obtain that    there is a constant $C(\boldsymbol \eta)>0$ such that
\begin{align}
\label{logest}
\norm{F^{\pm\eta}(z_{\pm\epsilon})}\leq\norm{F^{\pm\boldsymbol \eta}(z_{\pm\epsilon})}  + C \left| \int^{\pm \eta}_{\pm \boldsymbol \eta}\mathrm{d}\tilde \eta/\tilde \eta \right|
\leq 
 C(\boldsymbol \eta)  \left| \log \eta \right|   .
\end{align}
Inserting this in  Lemma \ref{keyest},  we obtain   
\begin{align}
\label{prelimdifferential}
\norm{\mathrm{d}/\mathrm{d}\eta \, F^{\pm\eta}(z_{\pm\epsilon})}\leq C(\boldsymbol \eta) \eta^{-1/2}  \left| \log \eta \right|   ,
\end{align}
and similarly as above, we find 
\begin{align}
\norm{F^{\pm\eta}(z_{\pm\epsilon})}\leq\norm{F^{\pm\boldsymbol \eta}(z_{\pm\epsilon})}  + C(\boldsymbol \eta) \left| \int^{\pm \eta}_{\pm \boldsymbol \eta}\mathrm{d}\tilde \eta\, \eta^{-1/2} \left| \log \eta \right|    \right| .
\end{align}
 We conclude that there is 
a constant $C(\boldsymbol \eta) >0$ such that
\begin{align}
\label{boundedF}
\norm{F^{\pm\eta}(z_{\pm\epsilon})}\leq C(\boldsymbol \eta) .
\end{align}
 Now we use the text below 
\eqref{DefiHoP} and take the limit $\eta\to 0^+$ in \eqref{boundedF}. 
  We conclude that \eqref{sup1} holds true (also \eqref{sup2}, taking $g=0$). 
   Analogously, we show \eqref{sup2}.
\end{proof}
In the remainder of this section we prove Proposition \ref{prop:absenceofacspecclose} (iii). The spirit of the proof is similar to the proof of statement (ii), however, we need additional estimates which are collected in the lemma below.  
\begin{lemma}
\label{keyest1}
For $g  \geq 0 ,\boldsymbol \eta>0$ sufficiently small, 
 $\eta\in (0,\boldsymbol \eta)$,  $\epsilon\in (0,1)$, $z, z'\in I$ and $z_{\pm\epsilon}:=z\pm i\epsilon$,   the following estimates hold true  
\begin{enumerate}
\item[(i)] \begin{align}
\norm{\mathrm{d}/\mathrm{d}\eta \, \left(   F^{\pm\eta}(z_{\pm\epsilon}) - F^{\pm\eta}(z'_{\pm\epsilon})     \right)}\leq C\eta^{-1/2}
\end{align}
\item[(ii)] \begin{align}
\norm{\mathrm{d}/\mathrm{d}\eta \, \left(   F^{\pm\eta}(z_{\pm\epsilon}) - F^{\pm\eta}(z'_{\pm\epsilon})     \right)}\leq C\eta^{-3/2}|z-z'|
\end{align}
\item[(iii)] \begin{align}
\norm{\mathrm{d}/\mathrm{d}\eta \, \left( F^{\pm\eta}(z_{\pm\epsilon})       -  F_0^{\pm\eta}(z_{\pm\epsilon})   \right)}\leq C\eta^{-3/2}g,
\end{align}
 see  Remark \ref{zero}.  
\end{enumerate}
\end{lemma}
\begin{proof}
\begin{enumerate}
\item[(i)]   It  follows from Lemma \ref{keyest} and  \eqref{boundedF}. 
\item[(iii)] 
Using the second resolvent identity,  Remark \ref{zero}, Remark  \ref{DEfiCommut} and \eqref{MOURRE}, we get  
\begin{align}
\mp i\frac{\mathrm{d}}{\mathrm{d}\eta}  \left(F^{\pm\eta}(z_{\pm\epsilon}) -F_0^{\pm\eta}(z_{\pm\epsilon}) \right)
&=\pm ig \frac{\mathrm{d}}{\mathrm{d}\eta}    \left( \rho R^{\pm\eta}(z_{\pm\epsilon})  \tilde V_{\eta}  R_0^{\pm\eta}(z_{\pm\epsilon})  \rho\right),
\end{align} 
where  (see \eqref{pl1})    
\begin{align}\label{Vtilde}
\tilde V_\eta := \sigma_1  \left(  \Phi_{\overline P}(f) \mp i \eta \chi(H_{\overline P})  \Phi_{ \overline P } (Df)  \chi(H_{\overline P}) \right).
\end{align} 
We write 
\begin{align}
\label{3qsplit'}
\mp i\frac{\mathrm{d}}{\mathrm{d}\eta}  \left(F^{\pm\eta}(z_{\pm\epsilon}) -F_0^{\pm\eta}(z_{\pm\epsilon}) \right) =g\left( W^{(1)}+W^{(2)}+W^{(3)}\right),
\end{align}
 where 
\begin{align}
&W^{(1)}:=\rho \left( \pm i \frac{\mathrm{d}}{\mathrm{d}\eta}   R^{\pm\eta}(z_{\pm\epsilon})    \right)\tilde V_{\eta}   R_0^{\pm\eta}(z_{\pm\epsilon})  \rho,
\\
&W^{(2)}:=\rho  R^{\pm\eta}(z_{\pm\epsilon}) \tilde V_{\eta}   \left( \pm i \frac{\mathrm{d}}{\mathrm{d}\eta}  R_0^{\pm\eta}(z_{\pm\epsilon})  \right)  \rho ,
\\
&W^{(3)}:=\rho  R^{\pm\eta}(z_{\pm\epsilon})     \chi(H_{\overline P})     
\Phi_{\overline P}(Df)      \chi(H_{\overline P})    R_0^{\pm\eta}(z_{\pm\epsilon}) \rho .
\end{align}
  Eqs.\ \eqref{Teo5i} and \eqref{MOURRE}  yield that
\begin{align}
\label{3qsplit}
W:=W^{(1)}+W^{(2)}&=\sum^4_{i=1}(W^{(1)}_i+W^{(2)}_i) ,
\end{align}
where 
\begin{align}
W^{(1)}_1:&=- \rho  R^{\pm\eta}(z_{\pm\epsilon})  [H_{\overline P}, i   \mathrm{d}\Gamma(D)]^0     R^{\pm\eta}(z_{\pm\epsilon}) \tilde V_{\eta}   R_0^{\pm\eta}(z_{\pm\epsilon}) \rho
\\
W^{(1)}_2:&= \rho  R^{\pm\eta}(z_{\pm\epsilon}) [H_{\overline P}, i   \mathrm{d}\Gamma(D)]^0 \overline \chi(H_{\overline P})  R^{\pm\eta}(z_{\pm\epsilon}) \tilde V_{\eta}    R_0^{\pm\eta}(z_{\pm\epsilon}) \rho
\\
W^{(1)}_3:&= \rho R^{\pm\eta}(z_{\pm\epsilon})  \overline \chi(H_{\overline P})  [H_{\overline P}, i   \mathrm{d}\Gamma(D)]^0  R^{\pm\eta}(z_{\pm\epsilon}) \tilde V_{\eta}     R_0^{\pm\eta}(z_{\pm\epsilon}) \rho 
\\
W^{(1)}_4:&=- \rho  R^{\pm\eta}(z_{\pm\epsilon})  \overline \chi(H_{\overline P})    [H_{\overline P}, i   \mathrm{d}\Gamma(D)]^0 \overline  \chi(H_{\overline P}) R^{\pm\eta}(z_{\pm\epsilon})  \tilde V_{\eta}      R_0^{\pm\eta}(z_{\pm\epsilon}) \rho 
\\
W^{(2)}_1:&=- \rho  R^{\pm\eta}(z_{\pm\epsilon}) \tilde V_{\eta}     R_0^{\pm\eta}(z_{\pm\epsilon})  [H_{0,\overline P}, i   \mathrm{d}\Gamma(D)]^0     R_0^{\pm\eta}(z_{\pm\epsilon}) \rho
\\
W^{(2)}_2:&= \rho  R^{\pm\eta}(z_{\pm\epsilon})   \tilde V_{\eta}     R_0^{\pm\eta}(z_{\pm\epsilon}) [H_{0,\overline P}, i   \mathrm{d}\Gamma(D)]^0 \overline \chi(H_{0,\overline P}) R_0^{\pm\eta}(z_{\pm\epsilon}) \rho
\\
W^{(2)}_3:&= \rho R^{\pm\eta}(z_{\pm\epsilon}) \tilde V_{\eta}    R_0^{\pm\eta}(z_{\pm\epsilon})  \overline \chi(H_{0,\overline P})  [H_{0,\overline P}, i   \mathrm{d}\Gamma(D)]^0   R_0^{\pm\eta}(z_{\pm\epsilon}) \rho 
\\
W^{(2)}_4:&=- \rho  R^{\pm\eta}(z_{\pm\epsilon})\tilde V_{\eta}     R_0^{\pm\eta}(z_{\pm\epsilon})   \overline \chi(H_{0,\overline P})    [H_{0,\overline P}, i   \mathrm{d}\Gamma(D)]^0 \overline  \chi(H_{0,\overline P})   R_0^{\pm\eta}(z_{\pm\epsilon}) \rho .
\end{align}
We observe from \eqref{3qsplit'} that in order to complete the proof of statement (iii) it suffices to show that
\begin{align}
\label{sufficient}
\norm{W}\leq C\eta^{-3/2} \qquad \text{and} \qquad  \norm{W^{(3)}}\leq C\eta^{-3/2}  .
\end{align}
 It follows from   Proposition \ref{PTOPNDJO}, \eqref{Teo5v} and similar estimates that  
that 
\begin{align}
\norm{\tilde V_\eta R^{\pm\eta}(z_{\pm\epsilon}) \rho }
&\leq  \norm{ \tilde V_{\eta}  (H_{f,\overline P}+i)^{-1}}\norm{(H_{f,\overline P}+i) (H_{\overline P}+i)^{-1} } \norm{(H_{\overline P}+i) R^{\pm\eta}(z_{\pm\epsilon}) \rho }
\notag \\
&\leq C \left(1+\eta^{-1/2}\norm{ F^{\pm\eta}(z_{\pm\epsilon})}^{1/2}\right)   
\end{align}
and similarly, using the adjoint operator, we find
\begin{align}
 &\norm{\rho R^{\pm\eta}(z_{\pm\epsilon}) \tilde V_\eta }\leq C\left(1+\eta^{-1/2}\norm{ F^{\pm\eta}(z_{\pm\epsilon})}^{1/2}\right).
\end{align}
  This and   \eqref{boundedF}  imply that   
\begin{align}
\label{3crucialestimate}
\norm{\tilde V_\eta R^{\pm\eta}(z_{\pm\epsilon}) \rho } \leq C \eta^{-1/2}   \qquad \text{and}\qquad \norm{\rho R^{\pm\eta}(z_{\pm\epsilon}) \tilde V_\eta }\leq C\eta^{-1/2} .
\end{align}
   Using additionally   \eqref{crucial1234}, we get 
\begin{align}
\label{W12}
\norm{W^{(1)}_2}\leq \norm{\rho  R^{\pm\eta}(z_{\pm\epsilon}) }\norm{[H_{\overline P}, i   \mathrm{d}\Gamma(D)]^0 \overline \chi(H_{\overline P})    R^{\pm\eta}(z_{\pm\epsilon})  }\norm{ \tilde V_{\eta}R_0^{\pm\eta}(z_{\pm\epsilon}) \rho}\leq C\eta^{-1} .
\end{align}
 Eqs.\ \eqref{3crucialestimate},   \eqref{chibarR} and \eqref{TEo5iV},    and the fact that  $[H_{\overline P}, i   \mathrm{d}\Gamma(D)]^0$ is  $H_{\overline P}$-bounded  (see Remark \ref{DEfiCommut}) imply that  
\begin{align}
\label{W13}
\norm{W^{(1)}_3}\leq  \norm{ R^{\pm\eta}(z_{\pm\epsilon})  \overline \chi(H_{\overline P})}\norm{  [H_{\overline P}, i   \mathrm{d}\Gamma(D)]^0  R^{\pm\eta}(z_{\pm\epsilon}) }\norm{  \tilde V_{\eta} R_0^{\pm\eta}(z_{\pm\epsilon}) \rho }\leq C\eta^{-3/2}.
\end{align}
Moreover, we  obtain  from \eqref{3crucialestimate}, 
\eqref{chibarR}  and  \eqref{crucial1234}  that 
\begin{align}
\label{W14}
\norm{W^{(1)}_4}
&\leq \norm{ \rho  R^{\pm\eta}(z_{\pm\epsilon})  \overline \chi(H_{\overline P})}\norm{    [H_{\overline P}, i   \mathrm{d}\Gamma(D)]^0 \overline  \chi(H_{\overline P}) R^{\pm\eta}(z_{\pm\epsilon})   } \norm{ \tilde V_{\eta} R_0^{\pm\eta}(z_{\pm\epsilon}) \rho}
\notag \\
&\leq C   \eta^{-1/2}   .
\end{align}
Analogously,   we deduce  that
\begin{align}
\label{W2easy}
\norm{W^{(2)}_2},\norm{W^{(2)}_3},\norm{W^{(2)}_4}\leq C\eta^{-3/2}.
\end{align}
Next, we estimate the terms $W^{(1)}_1$ and $W^{(2)}_1$.  
For $\phi,\psi\in\mathcal D( \mathrm{d}\Gamma(D))\cap \mathcal D(H_{\overline P})$, we find 
\begin{align}
\label{3q1split'}
\left\langle \phi,  \left( W^{(1)}_{1}+ W^{(2)}_{1}\right) \psi \right\rangle &= A_1+A_2+A_3+A_4 
  ,
\end{align}
where 
\begin{align}\label{As}
 A_1:=& - \left\langle  \left(  H^{\mp \eta}_{\overline P} -z_{\mp \epsilon}  \right) R^{\mp\eta}(z_{\mp\epsilon}) \rho \phi,     i   \mathrm{d}\Gamma(D)     R^{\pm\eta}(z_{\pm\epsilon})  \tilde V_{\eta}   R_0^{\pm\eta}(z_{\pm\epsilon}) \rho \psi \right\rangle 
\notag \\
&+\left\langle \left( - i   \mathrm{d}\Gamma(D) \right) R^{\mp\eta}(z_{\mp\epsilon}) \rho \phi,         \left(  H^{\pm \eta}_{\overline P} -z_{\pm \epsilon}  \right) R^{\pm\eta}(z_{\pm\epsilon})  \tilde V_{\eta}   R_0^{\pm\eta}(z_{\pm\epsilon}) \rho \psi \right\rangle   \notag 
,
\\
\notag  A_2 :=&\mp i\eta \bigg( \left\langle  M^2 R^{\mp\eta}(z_{\mp\epsilon})  \rho \phi,    i   \mathrm{d}\Gamma(D)    R^{\pm\eta}(z_{\pm\epsilon})   \tilde V_{\eta}   R_0^{\pm\eta}(z_{\pm\epsilon}) \rho \psi  \right\rangle
\notag \\ \notag
&- \left\langle  \left( - i   \mathrm{d}\Gamma(D) \right) R^{\mp\eta}(z_{\mp\epsilon})  \rho \phi,   M^2    R^{\pm\eta}(z_{\pm\epsilon}) \tilde V_{\eta}   R_0^{\pm\eta}(z_{\pm\epsilon}) \rho \psi \right\rangle \bigg),
\\
\notag A_3 :=& - \left\langle \left( H^{\mp \eta}_{0,\overline P} -z_{\mp \epsilon} \right) R_0^{\mp\eta}(z_{\mp\epsilon})    (\tilde V_{\eta})^*   R^{\mp\eta}(z_{\mp\epsilon}) \rho \phi ,    i   \mathrm{d}\Gamma(D)        R_0^{\pm\eta}(z_{\pm\epsilon}) \rho \psi \right\rangle 
\notag \\ \notag
&+ \left\langle \left( - i   \mathrm{d}\Gamma(D) \right) R_0^{\mp\eta}(z_{\mp\epsilon}) (\tilde V_{\eta})^*   R^{\mp\eta}(z_{\mp\epsilon}) \rho \phi ,   \left( H^{\pm \eta}_{0,\overline P} -z_{\pm \epsilon}       \right) R_0^{\pm\eta}(z_{\pm\epsilon}) \rho \psi \right\rangle 
,
\\
A_4 :=&\mp i\eta \bigg( \left\langle  M^2 R_0^{\mp\eta}(z_{\mp\epsilon}) (\tilde V_{\eta})^*   R^{\mp\eta}(z_{\mp\epsilon}) \rho \phi  ,    i   \mathrm{d}\Gamma(D)       R^{\pm\eta}(z_{\pm\epsilon}) \rho\psi \right\rangle
\notag\\
&-\left\langle  \left( - i   \mathrm{d}\Gamma(D) \right)R_0^{\mp\eta}(z_{\mp\epsilon})  (\tilde V_{\eta})^*   R^{\mp\eta}(z_{\mp\epsilon}) \rho \phi  ,   M^2      R^{\pm\eta}(z_{\pm\epsilon}) \rho\psi \right\rangle
\bigg)
  .
\end{align}
This is possible because  $\rho$ maps the Hilbert space $\mathcal{H}$ into the domain of $\mathrm{d}\Gamma(D)$  and --  by Lemma \ref{Invariance} --   $R^\pm( z_\pm\epsilon ), \:   (\tilde V_{\eta})^*   R^{\mp\eta}(z_{\mp\epsilon})  $  and  $   V_{\eta}   R^{\pm\eta}(z_{\pm\epsilon})    $  preserve   the domain of $\mathrm{d}\Gamma(D)$ (see above   \eqref{ComutGammaM2} --  this holds true also for $g=0$, see Remark \ref{zero}).  
We estimate
\begin{align}
 |A_2|  &\leq \eta \norm{\phi} \norm{\psi}\norm{\tilde V_{\eta}   R_0^{\pm\eta}(z_{\pm\epsilon}) \rho }\norm{R^{\pm\eta}(z_{\pm\epsilon}) }\norm{  R^{\mp\eta}(z_{\mp\epsilon})\rho }  \norm{ [M^2,  \mathrm{d}\Gamma(D)]^0  }
  .
\end{align}
 Eqs.\ \eqref{3crucialestimate},  \eqref{boundedF}, \eqref{Teo5v} and \eqref{TEo5iV}   imply that  
\begin{align}
\label{Wno}
|A_2| &\leq C\norm{\phi} \norm{\psi} \eta^{-1}
  ,
\end{align}
and analogously, we find
\begin{align}
\label{W'}
 |A_4| &\leq C\norm{\phi} \norm{\psi} \eta^{-1}
  .
\end{align}
As we argue above,    Lemma \ref{Invariance}  implies that    $R^\pm( z_\pm\epsilon ), \:   (\tilde V_{\eta})^*   R^{\mp\eta}(z_{\mp\epsilon})  $  and  $   V_{\eta}   R^{\pm\eta}(z_{\pm\epsilon})    $  preserve  the domain of $\mathrm{d}\Gamma(D)$  (see above   \eqref{ComutGammaM2} -  this holds true also for $g=0$, see Remark \ref{zero}).    
 Moreover,   the quadratic form  $[i \mathrm{d}\Gamma(D) ,\tilde V_{\eta} ]$ is represented by a $H_{\overline P}$-bounded operator that we denote by $[i \mathrm{d}\Gamma(D) ,\tilde V_{\eta} ]^0$ (see Lemma \ref{commM}). We obtain that  
\begin{align}
\label{split101}
A_1+A_3
=  & - \left\langle \left(-  i   \mathrm{d}\Gamma(D)   \right)  \rho \phi,      R^{\pm\eta}(z_{\pm\epsilon})\tilde V_{\eta}   R_0^{\pm\eta}(z_{\pm\epsilon}) \rho \psi\right\rangle 
\notag \\
&+ \left\langle  R^{\mp\eta}(z_{\mp\epsilon}) \rho  \phi, [\left(  i   \mathrm{d}\Gamma(D)   \right),  \tilde V_{\eta} ]^0  R_0^{\pm\eta}(z_{\pm\epsilon}) \rho \psi \right\rangle 
\notag  \\
&+ \left\langle   R_0^{\mp\eta}(z_{\mp\epsilon})  (\tilde V_{\eta})^*   R^{\mp\eta}(z_{\mp\epsilon}) \rho \phi  ,    i   \mathrm{d}\Gamma(D)  \rho \psi \right\rangle 
.
\end{align}   
It follows from \eqref{Teo5v},  \eqref{boundedF} and the fact that $[i \mathrm{d}\Gamma(D) ,\tilde V_{\eta} ]^0$ is 
 $H_{\overline P}$-bounded (see Lemma \ref{commM}) that 
\begin{align}\label{midle}
\Big | \left\langle  R^{\mp\eta}(z_{\mp\epsilon}) \rho  \phi, [\left(  i   \mathrm{d}\Gamma(D)   \right),  \tilde V_{\eta} ]^0  R_0^{\pm\eta}(z_{\pm\epsilon}) \rho \psi \right\rangle  \Big | \leq  \norm{\phi}\norm{\psi}\eta^{-1}.
\end{align}
We   obtain  from \eqref{3crucialestimate}  and \eqref{TEo5iV} that  
\begin{align}
 &\left|  \left\langle \left(-  i   \mathrm{d}\Gamma(D)   \right)  \rho \phi,      R^{\pm\eta}(z_{\pm\epsilon})  \tilde V_{\eta} R_0^{\pm\eta}(z_{\pm\epsilon}) \rho \psi \right\rangle    \right|
\leq C \norm{\phi}\norm{\psi}\eta^{-3/2},
\\
& \left| \left\langle   R_0^{\mp\eta}(z_{\mp\epsilon}) \tilde V_{\eta}^*    R^{\mp\eta}(z_{\mp\epsilon})  \rho  \phi,    i   \mathrm{d}\Gamma(D)  \rho \psi \right\rangle  \right|
\leq C \norm{\phi}\norm{\psi}\eta^{-3/2}.
\end{align}
This together with \eqref{split101} and  \eqref{midle}  yield  that 
\begin{align}
\label{Wsum'}
 |A_1+A_3| \leq C \norm{\phi}\norm{\psi}\eta^{-3/2}.
\end{align}
 It follows from \eqref{Wsum'}, \eqref{Wno}, \eqref{W'} and \eqref{3q1split'} that 
\begin{align}
\label{Wsum}
\norm{W^{(1)}_{1}+W^{(2)}_{1}}\leq C\eta^{-3/2}.
\end{align}
Collecting \eqref{3qsplit},  \eqref{W12}, \eqref{W13}, \eqref{W14}, \eqref{W2easy} and   \eqref{Wsum}, we   deduce  that 
\begin{align}
\label{sufficient1}
\norm{W}\leq C\eta^{-3/2}.
\end{align}
 Eqs.\ \eqref{boundedF} and \eqref{Teo5v}  together with the $H_{0,\overline P}$-boundedness of $\Phi_{\overline{P}}(Df)$ yield that   
\begin{align}
\norm{W^{(3)}}\leq C\eta^{-1}.
\end{align}
This together with \eqref{sufficient1} imply that \eqref{sufficient} holds true and, thereby, we  complete the proof of Item  (iii).
\item[(ii)]
The proof of Item  (ii) follows the same line of arguments as the proof of Item (iii). In fact, it is simpler since the term $\tilde V_\eta$ does not appear.
\end{enumerate}
\end{proof}
\begin{proof}[Proof of Proposition \ref{prop:absenceofacspecclose} (iii)]
   We estimate, for $z,z'\in I $,  
\begin{align}
\label{steps}
\norm{F^{0}(z'_{\pm\epsilon})- F^{0}_0(z_{\pm\epsilon})}\leq \norm{F^{0}(z'_{\pm\epsilon})- F^{0}(z_{\pm\epsilon}) }+  \norm{F^{0}(z_{\pm\epsilon})- F^{0}_0(z_{\pm\epsilon}) } .
\end{align}
Hence, it suffices to show that 
\begin{align}
\label{stepz}
 \norm{F^{0}(z'_{\pm\epsilon})- F^{0}(z_{\pm\epsilon}) }\leq C |z-z'|^{1/2},
\end{align}
and 
\begin{align}
\label{stepg}
 \norm{F^{0}(z_{\pm\epsilon})- F_0^{0}(z_{\pm\epsilon}) }\leq C g^{1/2}.
\end{align}
 In the remainder of the proof we show \eqref{stepz} and \eqref{stepg}. We start with the first estimate and obtain for $\tilde \eta \in (0,\boldsymbol \eta)$
 \begin{align}
\label{continz}
& F^{0}(z'_{\pm\epsilon})-F^{0}(z_{\pm\epsilon})  
=-\int^{\tilde \eta}_{0}\mathrm{d}\eta \, \frac{d}{d\eta} ( F^{\eta}(z'_{\pm\epsilon})-F^{\eta}(z_{\pm\epsilon})  )
\notag \\
&-\int_{\tilde \eta}^{\boldsymbol \eta}\mathrm{d}\eta \, \frac{d}{d\eta} ( F^{\eta}(z'_{\pm\epsilon})-F^{\eta}(z_{\pm\epsilon})  )
+ F^{\boldsymbol \eta}(z'_{\pm\epsilon})-F^{\boldsymbol \eta}(z_{\pm\epsilon})  .
\end{align}
It follows from Lemma \ref{keyest1} (i) that 
\begin{align}
\label{zo}
\norm{\int^{\tilde \eta}_{0}\mathrm{d}\eta \, \frac{d}{d\eta} ( F^{\eta}(z'_{\pm\epsilon})-F^{\eta}(z_{\pm\epsilon})  )}\leq C \tilde \eta^{1/2}.
\end{align}
Moreover, it follows from  Lemma \ref{keyest1} (ii) that 
\begin{align}
\label{z1}
\norm{\int_{\tilde \eta}^{\boldsymbol \eta}\mathrm{d}\eta \, \frac{d}{d\eta} ( F^{\eta}(z'_{\pm\epsilon})-F^{\eta}(z_{\pm\epsilon})  )}\leq C |z-z'| \tilde \eta^{-1/2} ,
\end{align}
and it follows from the resolvent identity   that  there is a constant $C(\boldsymbol \eta)>0$ such that 
\begin{align}
\label{z2}
\norm{ F^{\boldsymbol \eta}(z'_{\pm\epsilon})-F^{\boldsymbol \eta}(z_{\pm\epsilon}) } \leq C(\boldsymbol \eta ) |z-z'| .
\end{align}  
Note that, in principle, the constant $C(\boldsymbol \eta )$ could depend on $\epsilon$ and $z$. However, this is not the case,  see \eqref{TEo5iV}   
Choosing $\tilde \eta =|z-z'|^{1/2} $, we  get   \eqref{stepz} from 
    \eqref{continz} -- \eqref{z2}.   Eq.\  \eqref{stepg} can be proven analogously  employing  item (iii) of  Lemma \ref{keyest1} instead of   item   (ii). 
\end{proof}

\begin{appendix}

\section{Construction of the ground state}
\label{app:gs-proof}
In this chapter, we construct the ground  state of the Hamiltonian $H$ and provide a proof for Proposition \ref{prop:gs}.
First of all, for $0<r<r'<\infty$ and $w\in\C$, we introduce the notation for the open   annulus  in the complex plane: 
\begin{align}
\label{def:disc}
D(r,r',w)=\left\{z\in \C : r<|z-w|<r' \right\} .
\end{align}
\begin{lemma}
\label{lemma:resestneargs}
Let $g>0$ be small enough. Then,  $H-z$ is invertible for all  $z\in \overline{D(m/4,m/2,0)}$  (defined in \eqref{def:disc}) and
\begin{align}
\norm{(H-z)^{-1}}\leq 2 \norm{(H_0-z)^{-1}}\leq 8/m \qquad \forall z\in  \overline{D(m/4,m/2,0)}  .
\end{align}
\end{lemma}
\begin{proof}
First of all, note that $\sigma(H_0)=\{0\}\cup [m,\infty)$. This implies that 
\begin{align}
\text{dist}\left(   D(m/4,m/2,0) , \sigma(H_0) \right)\geq 4/m, 
\end{align}
and hence, $H_0-z$ is invertible for all $z\in D(m/4,m/2,0) $, and for those $z$, we have
\begin{align}
\label{eq:stand+}
\norm{(H_0-z)^{-1}}\leq 4/m.
\end{align}
Moreover, it follows from the standard estimate in  Proposition~\ref{PTOPNDJO}   that
\begin{align}
\norm{V(H_0+1)^{-1}}\leq C,
\end{align}
and hence,  we obtain for all $z\in D(m/4,m/2,0) $
\begin{align}
\norm{V(H_0-z)^{-1}}\leq\norm{V(H_0+1)^{-1}}\norm{\frac{H_0+1}{H_0-z}}\leq C \sup_{y\geq 0}\left|\frac{y+1}{y-z}\right| \leq C(3+4/m) .
\end{align}
Consequently, for $g>0$ sufficiently small, we   find 
\begin{align}
\label{C7}
\norm{V(H_0-z)^{-1}} \leq Cg   \leq  1/2 ,
\end{align}
and hence, 
\begin{align}
H-z=(1+gV(H_0-z)^{-1})(H_0-z)
\end{align}
is invertible 
for all  $z\in D(m/4,m/2,0) $ and the resolvent fulfills
\begin{align}
\norm{(H-z)^{-1}}\leq 2 \norm{(H_0-z)^{-1}}\leq 8/m .
\end{align}
\end{proof}
\begin{definition}
\label{def:proj}
We define the contour 
\begin{align}
\zeta: [0,2\pi] \to \C, \qquad \varphi\mapsto \zeta(t):= m/4 e^{it} .
\end{align}  
Furthermore, we define the projections 
\begin{align}
\label{eq:projat}
P_{0,\text{at}}:= (-2\pi i)^{-1} \oint_{\zeta} \mathrm{d}z \,(H_0-z)^{-1}= P_{\varphi_0} \otimes P_\Omega
\end{align}
and
\begin{align}
P_{0}:= (-2\pi i)^{-1} \oint_{\zeta} \mathrm{d}z \, (H-z)^{-1}.
\end{align}
Here, $P_{\varphi_0}$  denotes the projection onto $\varphi_0$ and $P_{\Omega}$  the projection onto the vacuum $\Omega\in\mathcal F[\mathfrak{h}]$. The  equality in \eqref{eq:projat} can be seen by a direct calculation.
\end{definition}
\begin{lemma}
\label{lemma:proj}
Let $g>0$ be small enough and Assumption \ref{as} hold true. Then, we find 
\begin{align}
\norm{P_{0}-P_{0,\text{at}}}\leq gC< 1 .
\end{align}
\end{lemma}
\begin{proof}
It follows from Definition \ref{def:proj} that
\begin{align}
\norm{P_{0}-P_{0,\text{at}}}&\leq (2\pi )^{-1} \int_{0}^{2\pi} \mathrm{d}t\, \norm{ (H-m/4e^{it})^{-1}- (H_0-m/4e^{it})^{-1}}
\notag \\
 &\leq g\sup_{t\in [0,2\pi]}\norm{ (H-m/4e^{it})^{-1}}\norm{V (H_0-m/4e^{it})^{-1}}
,
\end{align}
where we used the resolvent identity in the second step. This together with  \eqref{C7}  and Lemma \ref{lemma:resestneargs} completes the proof.
\end{proof}
\begin{proof}[Proof of Proposition \ref{prop:gs}]
Clearly, $P_{0,\text{at}}=P_{\varphi_0} \otimes P_\Omega$ is a rank-one projection, and hence, it follows from Lemma \ref{lemma:proj} that also $P_0$ is a rank-one projection.
 Consequently, the self-adjoint operator $H$ has exactly one  eigenvalue in $(-m/4,m/4)$ which we call $\lambda_0$ and $\Psi_{\lambda_0}:=P_0 \varphi_0\otimes \Omega\in \mathcal H$ is  nonzero   and fulfills $H\Psi_{\lambda_0}=\lambda_0\Psi_{\lambda_0}$. 
 
In the remainder of the proof we compute $\lambda_0$ up to second order in $g$.
\begin{align}
\label{calc1}
(\lambda_0-e_0)  \left|   \left\langle \phi_0,  P_0 \phi_0\right\rangle \right|   =\left\langle \phi_0, (H-H_0) P_0 \phi_0\right\rangle=g\left\langle \phi_0, V P_0 \phi_0\right\rangle ,
\end{align}
where we have introduced the notation $\phi_i=\varphi_i\otimes \Omega$ for $i=0,1$.
Moreover, the resolvent identity yields that
\begin{align}
&\left\langle \phi_0,V(H-z)^{-1}\phi_0\right\rangle =
\left\langle \phi_0, V(H_0-z)^{-1}\phi_0\right\rangle -g \left\langle \phi_0,V(H_0-z)^{-1}V(H_0-z)^{-1}\phi_0\right\rangle
\notag \\
&+g^2\left\langle \phi_0,V(H_0-z)^{-1}V(H_0-z)^{-1}V(H_0-z)^{-1} \phi_0\right\rangle
\notag \\
&-g^3\left\langle \phi_0,V(H-z)^{-1}V(H_0-z)^{-1}V(H_0-z)^{-1} V(H_0-z)^{-1} \phi_0\right\rangle .
\end{align}
Note that the even orders of $g$ vanish due to symmetry and recall from \eqref{eq:stand+} that   $\norm{V(H_0-z)^{-1}}\leq C$. This implies that
\begin{align}
\norm{V(H-z)^{-1}}\leq   \norm{(H_0-z)(H-z)^{-1}}  \norm{V(H_0-z)^{-1}}\leq C (1+gC) .
\end{align}
Consequently, we obtain
\begin{align}
&\left\langle \phi_0,V(H-z)^{-1}\phi_0\right\rangle =
 -g (e_0-z)^{-1}\left\langle a(f)^*\phi_1,(H_0-z)^{-1}a(f)^*\phi_1\right\rangle
+\tilde R_0(g) 
\notag \\
&= -g (e_0-z)^{-1}  \int \mathrm{d}^3k\, |f(k)|^2 (e_1+\omega(k)-z)^{-1}+\tilde R_0(g) 
 ,
\end{align}
where  $|\tilde R_0(g)|\leq Cg^3$. Then, it follows from \eqref{calc1} together with Definition \ref{def:proj} that
\begin{align}
\lambda_0=e_0-g^2\Gamma_0+R_0(g)
\end{align}
where $R_0(g)=g   \left|   \left\langle \phi_0,  P_0 \phi_0\right\rangle \right|^{-1}  \tilde R_0(g)$ and 
\begin{align}
\Gamma_0:=(-2\pi i)^{-1} \oint_\zeta\mathrm{d}z \, (e_0-z)^{-1} \int \mathrm{d}^3k\, |f(k)|^2 (e_1+\omega(k)-z)^{-1} .
\end{align}
Fubini's theorem allows for interchanging the order of integration, and hence, we obtain from the Cauchy integral theorem 
\begin{align}
\Gamma_0= \int \mathrm{d}^3k\, |f(k)|^2 (e_1-e_0+\omega(k))^{-1} .
\end{align}
This completes the proof of the first part of the proposition. The second part follows from the definition of the ground state:
\begin{align}
\Psi_{\lambda_0}=P_0 \varphi_0\otimes \Omega = \varphi_0\otimes \Omega + \widetilde \Psi_{\lambda_0},
\end{align}
where $\widetilde \Psi_{\lambda_0}:= (P_0-P_{0,\text{at}}) \varphi_0\otimes \Omega \in \mathcal H$
and Lemma \ref{lemma:proj} yields that $\norm{\widetilde \Psi_{\lambda_0}}\leq Cg$.
Moreover, note that $\varphi_0\otimes \Omega$ is the unique ground state of $H_0$, and hence, $P_{0,\text{at}}$ is a rank-one projector. We conclude 
the uniqueness  of $\Psi_{\lambda_0}$ again from Lemma \ref{lemma:proj}.
\end{proof}

\section{Spectral projections}
\label{app:specproj}

\begin{definition}
\label{def:almostana}
For $\upsilon \in \mathit C^{\infty}(\R,\C)$, we define its almost analytic extension by
\begin{align}
\tilde \upsilon :\C \to \C, \qquad
\tilde \upsilon(z) = \sigma(\Re z,\Im z) \sum^n_{r=0} \frac{(i\Im z)^r}{r!}\upsilon^{(r)}(\Re z) ,
\end{align}
where $n\in\N$, $\upsilon^{(r)}$ denotes the $r$-th derivative of $\upsilon$ and 
\begin{align}
\sigma(\Re z,\Im z) := \tau\left(\frac{\Im z}{\sqrt{(\Re z)^2 +1}}\right)
\end{align} 
for some $\tau\in \mathit C^\infty(\R,\C)$ with $\tau(t)=1$ for all $|t|<1$ and $\tau(t)=0$ for all $|t|>2$.
It follows from \cite[Section 2.2]{Davies2} that
\begin{enumerate}
\item[(i)] $\tilde \upsilon$ is smooth as a function of $(\Re z, \Im z)$. 
\item[(ii)]   If $v$ is compactly supported, 
 $|\partial_{\overline z} \,  \tilde \upsilon   (z) |\leq C |\Im z|^n$ (where 
 $\frac{d}{d \overline z} = \frac{1}{2} ( \frac{d}{d  x} + i\frac{d}{d y  }) $, with $z = x+iy$).
\end{enumerate}
\end{definition}
\begin{theorem}[Helffer-Sj\"ostrand formula]
\label{thm:hsf}
For every selfadjoint operator and any  $\upsilon \in \mathit C_0^{\infty}(\R,\C)$, the next formula holds true 
\begin{align}\label{HS}
\upsilon(O) = \pi^{-1} \int_\C \mathrm{d}x \mathrm{d}y \, \partial_{\overline z} \,  \tilde \upsilon (z) (O-z)^{-1},
\end{align}
where  $z=x+iy$, for $x,y\in \R$. Eq.\ \eqref{HS} does not depend on $n$ and $\sigma$. 
\end{theorem}
\begin{proof}[Proof of Lemma \ref{lemma:specproj}]
We only prove  \eqref{eq:specproj'}. 
Since   \eqref{HS} does not depend on $\sigma$, we choose $n = 2$ and, for $s >0$,  
 $ \sigma_s(\Re z,\Im z) := \tau\left( \frac{1}{s} \frac{\Im z}{\sqrt{(\Re z)^2 +1}}\right)   $. We denote by $\tilde \chi_s $ the corresponding almost analytic extension of $\chi_s$. It follows form \eqref{HS} and the resolvent equation that 
\begin{align}\label{verg1}
\| \chi_s(H) -  \chi_s(H_0)\| = \pi^{-1} \Big \| \int_\C \mathrm{d}x \mathrm{d}y \, \partial_{\overline z} \,  \tilde \chi_s (z) (H-z)^{-1} g V    (H_0-z)^{-1}       \Big \|.
\end{align} 
We calculate now  
\begin{align}\label{davies}
\partial_{\overline z} \,  \tilde \chi_s (z) = \frac{1}{2} \sum_{r = 0}^2
 \chi_s^{(r)}(x) (iy)^r/r! (\frac{\partial}{x}\sigma_s +  i \frac{\partial}{\partial y} \sigma_s  )  +   
\frac{1}{2} \chi_s^{(n+1)}(x) (iy)^n/n! \sigma . 
\end{align}
Notice that 
  $ \|  (H-z)^{-1} g V  \frac{1}{H_f + 1} (H_f + 1) (H_0-z)^{-1} \|  \leq  C g \frac{1}{|y|^2}    $. Moreover,   $  | \chi_s^{(r)}(x) | \, |y|^r |\frac{\partial}{x}\sigma_s +  i \frac{\partial}{\partial y} \sigma_s  | \leq C \frac{1}{s}  \frac{|y|^r}{s^r}   $,
  $| \chi_s^{(n+1)}(x)| |y|^n |\sigma | \leq C \frac{1}{s^3}  |y|^2   $. This together with 
  \eqref{verg1} yields
\begin{align}
\| \chi_s(H) -  \chi_s(H_0)\|  &\leq C g  \sum_{r = 0}^2
\int_{   \supp(  | \chi_s^{(r)} | \,  |\frac{\partial}{x}\sigma_s +  i \frac{\partial}{\partial y} \sigma_s  | )  } \mathrm{d}x \mathrm{d}y\,   \frac{1}{s}  \frac{|y|^{r-2}}{s^r}    
\notag \\
&+C g \int_{   \supp(  | \chi_s^{(3)} | |\sigma|  )}\mathrm{d}x \mathrm{d}y\,  \frac{1}{s^3}  .  
\end{align}  
For $y\in \R$, we observe that  the diameter of the support of the functions $\R\ni x\mapsto   | \chi_s^{(r)}(x) | \,  |\frac{\partial}{x}\sigma_s(x,y) +  i \frac{\partial}{\partial y} \sigma_s(x,y)  | $ and $\R\ni x\mapsto \supp(  | \chi_s^{( 3  )}(x)|  |\sigma(x,y) |) $ is of order $s$. Moreover, for $x\in\R$, we find that  the diameter of the support of the function $\R\ni y\mapsto \supp(  | \chi_s^{( 3  )}(x)|  |\sigma(x,y) |) $ is of order $s$. We conclude that
\begin{align}
\| \chi_s(H) -  \chi_s(H_0)\| \leq C \frac{g}{s}  ,
\end{align}   
which is the desired result. 
\end{proof}

\section{Standard Results from Mourre Theory}\label{Standard}

 In this section we prove all assertions and estimates described at the beginning of Section  \ref{app:limab}, upto   \eqref{Teo5v}.  We adapt the
proofs of \cite{cycon} to our model.   
\begin{lemma}
\label{mainlemma}
Recall  $\chi\in C^{\infty}_c(\R,[0,1])$ from Definition \ref{def.chi}.
For $g   \geq 0 ,\boldsymbol \eta>0$ sufficiently small, 
 $\eta\in (0,\boldsymbol \eta)$,  $\epsilon\in (0,1)$, $z\in I$ and $z_{\pm\epsilon}:=z\pm i\epsilon$, the following statements hold true:
\begin{enumerate}
\item[(i)] The operator $R^{\pm\eta}(z_{\pm\epsilon})$ introduced in  \eqref{DefiHoP}    exists and 
 it  is in $C^1((0,\boldsymbol \eta))$  and  $C^0([0,\boldsymbol \eta))$ with respect to $\eta$. Moreover, the following identity holds true:
\begin{align}
\label{diffeta}
\mathrm{d}/\mathrm{d}\eta \, R^{\pm\eta}(z_{\pm\epsilon})    =\pm iR^{\pm\eta}(z_{\pm\epsilon})  M^2     R^{\pm \eta}(z_{\pm\epsilon}) , \qquad \forall \eta\in (0,\boldsymbol \eta).
\end{align}
\item[(ii)] 
 \begin{align}
\norm{(H_{\overline P}+i) \chi(H_{\overline P})R^{\pm\eta}(z_{\pm\epsilon})\psi}\leq C\eta^{-1/2}\left| \left\langle   \psi,   R^{\pm\eta}(z_{\pm\epsilon})\psi\right\rangle \right|^{1/2}.
\end{align}
\item[(iii)]
\begin{align}
\norm{(H_{\overline P}+i)\overline \chi(H_{\overline P})R^{\pm\eta}(z_{\pm\epsilon})}\leq C,
\end{align}
where we   recall that   $\overline \chi=1-\chi$. 
\item[(iv)]  
\begin{align}
\norm{(H_{\overline P}+i)R^{\pm\eta}(z_{\pm\epsilon})}\leq C/\eta .
\end{align}
\item[(v)] 
 \begin{align}
\norm{  (H_{\overline P}+i)   R^{\pm\eta}(z_{\pm\epsilon})\rho}\leq C\left(1+\eta^{-1/2}\norm{ F^{\pm\eta}(z_{\pm\epsilon})}^{1/2}\right) .
\end{align}
\end{enumerate} 
   The  constants  $C$ above do  not depend on $\eta$, $\epsilon$, $z$ and $g$, see Remark \ref{rem:const}.  
\end{lemma}
\begin{proof}
\begin{enumerate}
\item[(i)] Recall that $H_{\overline P}$ is a closed operator and $M^2$ is 
 bounded  (see Remark \ref{DEfiCommut}). Consequently,  $H_{\overline{P}}^{\pm\eta}$ is  closed. For $\psi \in \mathcal D(H_{\overline P})$, we observe  that  
\begin{align}
\label{crucialepsilon}
\norm{\left(H^{\pm\eta}_{\overline P}  - z_{\pm \epsilon}\right)\psi}^2=\norm{\left(H^{\pm\eta}_{\overline P}  - z\right)\psi}^2 +\epsilon^2 \norm{\psi}^2+2\eta \epsilon\norm{M\psi}^2 ,
\end{align}
and, thereby, 
the range of $H^{\pm\eta}_{\overline P}  - z_{\pm \epsilon}$ is closed and $H^{\pm\eta}_{\overline P}  - z_{\pm \epsilon}$ is injective.  It also follows from the  equation above  that its inverse is   bounded.
Moreover, $\left(H^{\pm\eta}_{\overline P} - z_{\pm \epsilon}\right)^*$ fulfills a similar estimate  and it is, therefore, injective.  This implies    that the range of  $H^{\pm\eta}_{\overline P}  - z_{\pm \epsilon}$ is dense 
 and  because it is also closed,  $H^{\pm\eta}_{\overline P}  - z_{\pm \epsilon}$ is surjective.  

In addition,   the resolvent identity yields that 
\begin{align}
\label{different}
R^{\pm\eta}(z_{\pm\epsilon}) -R^{\pm\eta_0}(z_{\pm\epsilon})  =\pm i (\eta-\eta_0)R^{\pm\eta}(z_{\pm\epsilon})  M^2     R^{\pm \eta_0}(z_{\pm\epsilon})  .
\end{align} 
 It follows from \eqref{crucialepsilon} that there is a constant $C>0$ (independent of $\eta$) such that $\norm{R^{\pm\eta}(z_{\pm\epsilon})}\leq C/\epsilon$.
This together with  \eqref{different} and  the fact that $M^2$ is bounded   implies  that 
$ R^{\pm\eta}(z_{\pm\epsilon}) $ is continuous with respect to $\eta$, for $\eta \geq 0$, and differentiable for $\eta > 0$. Moreover, taking $\eta \to 0$ in \eqref{different} we get \eqref{diffeta}.   
\item[(ii)] 
It follows from Lemma \ref{lemma:comestclose} that there is a constant $\alpha>0$ such that for $\psi\in \mathcal H$
 \begin{align}
&\norm{(H_{\overline P}+i) \chi(H_{\overline P})R^{\pm\eta}(z_{\pm\epsilon})\psi}^2
 =  \left\langle \psi, R^{\pm \eta}(z_{\pm\epsilon})^* (H_{\overline P}^2+1) \chi^2(H_{\overline P})R^{\pm\eta}(z_{\pm\epsilon}) \psi  \right\rangle
\notag \\
& \leq((e_1+\delta)^2+1) \alpha^{-1}  \left\langle \psi, R^{\pm \eta}(z_{\pm\epsilon})^* \alpha \chi^2(H_{\overline P})R^{\pm\eta}(z_{\pm\epsilon}) \psi  \right\rangle
\notag \\
&\leq ((e_1+\delta)^2+1)   (2\alpha\eta)^{-1}  \left\langle \psi, R^{\pm\eta}(z_{\pm\epsilon})^* (2\eta M^2+2\epsilon)R^{\pm\eta}(z_{\pm\epsilon}) \psi  \right\rangle
\notag \\
&= ((e_1+\delta)^2+1) (2\alpha\eta)^{-1}  \left\langle \psi, i(R^{\pm\eta}(z_{\pm\epsilon})^*-R^{\pm\eta}(z_{\pm\epsilon})) \psi  \right\rangle
\notag \\
&\leq((e_1+\delta)^2+1) (\alpha\eta)^{-1} \left| \left\langle   \psi,   R^{\pm\eta}(z_{\pm\epsilon})\psi\right\rangle \right|.
\end{align}
This implies then statement (ii).
\item[(iii)]   We calculate  
\begin{align}
\label{iii}
\overline \chi(H_{\overline P})R^{\pm\eta}(z_{\pm\epsilon})&= \overline\chi(H_{\overline P})R^0(z_{\pm\epsilon}) \left(  H_{\overline{P}}-z_{\pm \epsilon}) R^{\pm\eta}(z_{\pm\epsilon}   \right)
\notag \\
&= \overline \chi(H_{\overline P})R^0(z_{\pm\epsilon}) \left(  1\pm i\eta M^2 R^{\pm\eta}(z_{\pm\epsilon})    \right).
\end{align}
 It follows from   Definition \ref{def.chi} and  \eqref{DefI} that    $ \norm{\overline \chi(H_{\overline P})R^0(z_{\pm\epsilon})}\leq 4/\delta$. Moreover, 
\begin{align}
 \norm{ H_{\overline P} \overline \chi(H_{\overline P})R^0(z_{\pm\epsilon}) }    =    \norm{ \overline \chi(H_{\overline P})     +z_{\pm\epsilon}\overline \chi(H_{\overline P})  R^0(z_{\pm\epsilon}) }   .
\end{align}
We  obtain that  
\begin{align}
 \norm{(H_{\overline P}+i)  \overline \chi(H_{\overline P})R^0(z_{\pm\epsilon}) } \leq C.
\end{align}
This together with \eqref{iii} and the boundedness of $M^2$ yields that
\begin{align}
\label{iv}
\norm{(H_{\overline P}+i)\overline \chi(H_{\overline P})R^{\pm\eta}(z_{\pm\epsilon})}&\leq C \left(  1+\eta  \norm{R^{\pm\eta}(z_{\pm\epsilon}) }   \right) .
\end{align}
Statement (iii) follows then by (iv) which is proven below.
\item[(iv)] It follows from (ii) together with \eqref{iv} that there are constants $C,\tilde C>0$ such that
\begin{align}
&1+\norm{(H_{\overline P}+i)R^{\pm\eta}(z_{\pm\epsilon})}
\notag \\
&\leq  1+ \norm{(H_{\overline P}+i)\overline \chi(H_{\overline P}) R^{\pm\eta}(z_{\pm\epsilon})}+ \norm{(H_{\overline P}+i)\chi(H_{\overline P}) R^{\pm\eta}(z_{\pm\epsilon})}
\notag \\
&\leq  1+\tilde C \left(  1+\eta  \norm{R^{\pm\eta}(z_{\pm\epsilon}) }   \right) +C\eta^{-1/2} \norm{R^{\pm\eta}(z_{\pm\epsilon}) }^{1/2}
.
\end{align}
We fix $\boldsymbol \eta>0$ sufficiently small such that $\tilde C\boldsymbol \eta\leq 1/2$ and $\tilde C +1\leq C\boldsymbol \eta^{-1/2}$. Then, employing $|x|+1\leq 2 \sqrt{x^2+1}$ for  all $x\in \R$, we conclude for $\eta\in (0,\boldsymbol \eta)$ 
\begin{align}
&1+\norm{(H_{\overline P}+i)R^{\pm\eta}(z_{\pm\epsilon})}
\leq C\eta^{-1/2} \left( 1+  \norm{R^{\pm\eta}(z_{\pm\epsilon}) }^{1/2}  \right) +   \frac12 \left(  1+  \norm{R^{\pm\eta}(z_{\pm\epsilon}) }   \right) 
\notag \\
&\leq 2C\eta^{-1/2} \left( 1+  \norm{(H_{\overline P}+i)R^{\pm\eta}(z_{\pm\epsilon}) } \right)^{1/2}  +   \frac12 \left(  1+  \norm{(H_{\overline P}+i)R^{\pm\eta}(z_{\pm\epsilon}) } \right)  .
\end{align}
This yields then
\begin{align}
1+\norm{(H_{\overline P}+i)R^{\pm\eta}(z_{\pm\epsilon})}\leq 4C\eta^{-1/2} \left(1+   \norm{(H_{\overline P}+i)R^{\pm\eta}(z_{\pm\epsilon}) }\right)^{1/2},
\end{align}
and hence,
\begin{align}
\norm{(H_{\overline P}+i)R^{\pm\eta}(z_{\pm\epsilon})}\leq 16C^2\eta^{-1} ,
\end{align}
which implies statement (iv).
\item[(v)] For $\psi \in \mathcal H$, we apply statement (ii) to the vector $\rho \psi \in \mathcal H$ and find that there is a  constant $C>0$ such that
\begin{align}
\norm{  (H_{\overline P}+i)  \chi(H_{\overline P})R^{\pm\eta}(z_{\pm\epsilon})\rho \psi}\leq C\eta^{-1/2}\left| \left\langle   \psi,   F^{\pm\eta}(z_{\pm\epsilon})\psi\right\rangle \right|^{1/2},
\end{align}
which implies
\begin{align}
\label{stat10011}
\norm{  (H_{\overline P}+i)  \chi(H_{\overline P})R^{\pm\eta}(z_{\pm\epsilon})\rho}\leq C\eta^{-1/2}\norm{ F^{\pm\eta}(z_{\pm\epsilon})}^{1/2}.
\end{align}
In addition, it follows from statement (iii) that 
\begin{align}
\norm{  (H_{\overline P}+i)  \overline \chi(H_{\overline P})R^{\pm\eta}(z_{\pm\epsilon})\rho}\leq C.
\end{align}
This together with \eqref{stat10011} completes the proof of statement (v).
\end{enumerate}
\end{proof}

\section{Domain Properties and Commutator Estimates in Mourre Theory}
\label{Nonstandard}

\subsection{Domain Properties in Mourre Theory  }

 In this section we prove auxiliary technical results that we need in Section 
\ref{app:limab}.  In particular, we prove that   $ R^{\pm \eta}(z_{\pm \epsilon}) $ (see \eqref{DefiHoP}) leaves  the domain of $ \mathrm{d}\Gamma(D)  $ invariant -- this (and similar results)  might be regarded as the main result of this section, see  Lemma  \ref{Invariance}.  In this paper we do not use the standard strategy and we believe that our method is much simpler and direct than the usual one:
A novelty of our presentation is that we do not employ the usual techniques to study domain problems and commutators.  The standard presentation of Mourre theory includes a scale of Hilbert spaces and a regularization of the generator of dilations in order to address domain problems (which is a technical and delicate issue --  see \cite{cycon}). In our case, instead of    stating scales of Hilbert spaces explicitly and regularizing the generator of dilations,  we directly dilate the operators at stake.           
We point out to the reader that the details of the arguments in this section are rarely found in the literature. 
A presentation of similar arguments may be found, e.g.,  in  \cite{fgs-spectral}.
\begin{definition}\label{DilatedOperator}
Let  $B$ be a closed operator, defined in $\mathcal H$.  For every $\beta  \in \mathbb{R}$,  we denote its dilation  by 
\begin{align}\label{dilation} B^{( \beta )} = e^{-i \beta  \mathrm{d} \Gamma (D)} B e^{i \beta \mathrm{d} \Gamma (D)}. 
\end{align}
For every function $ h: \mathbb{R}^3 \to \mathbb{R}  $ we denote by 
$ h^{(\beta)}(k) :=  h( e^{\beta} k ) $. A direct calculation shows that (see Definition \ref{def:secquant})
\begin{align}\label{Hdilated}
H_{\overline{P}}^{(\beta)} = H_{\overline P}(\omega^{(\beta)}, u_{\beta} f), 
\hspace{.3cm} (M^2)^{(\beta)} =  \chi(H_{\overline{P}}^{(\beta)})  
 H_{\overline P}(\xi^{(\beta)}, u_{\beta} D f)
 \chi(H_{\overline{P}}^{(\beta)}),
\end{align}
see Remark \ref{DEfiCommut}, and (see \eqref{DefiHoP})
\begin{align}\label{Rdilated}
(H^{\pm \eta }_{\overline{P} })^{(\beta)} := H_{\overline{P}}^{(\beta)}  \mp  i \eta (M^2)^{(\beta)},   \qquad (R^{\pm\eta}(z_{\pm\epsilon}))^{(\beta)} = \left( (H^{\pm\eta}_{\overline P})^{(\beta)}-z_{\pm\epsilon}  \right)^{-1}.
\end{align}
\end{definition}
\begin{lemma}\label{MAIN}
Let $B$ be a bounded operator in $\mathcal H$.  Assume that the map $ \beta \mapsto B^{(\beta)}$ is continuous at $0$ and, for every   
$\phi \in  \mathcal{D}(  \mathrm{d} \Gamma (D) )  $, the limit 
\begin{align}\label{F4}
\lim_{  \beta \to 0} \frac{1}{ \beta } (B^{(  \beta )} - B ) \phi  
\end{align}
exists.  Then,  $\mathcal{D}(  \mathrm{d} \Gamma (D) )$ is invariant under $B$. In particular this holds true if the map $ \beta \mapsto B^{(\beta)}$ is differentiable at $0$. 
\end{lemma}
\begin{proof} We recall that $B  \phi \in  \mathcal{D}(  \mathrm{d} \Gamma (D) )  $ if and only if the function $ \beta \mapsto   e^{ - i  \beta   \mathrm{d} \Gamma (D)   } B  \phi   $ is differentiable at $0$. 
Set $\phi \in \mathcal{D}(  \mathrm{d} \Gamma (D)  )  $.   We notice that the limit 
\begin{align}
\lim_{  \beta \to  0}\frac{1}{  \beta } (e^{-i  \beta   \mathrm{d} \Gamma (D) } - 1) B \phi & =  \lim_{  \beta  \to  0}\frac{1}{  \beta  }(B^{(  \beta )}  -  B )\phi + B^{(  \beta  )}
\frac{1}{ \beta  } (e^{-i \beta  \mathrm{d} \Gamma (D) } - 1)  \phi  
\end{align}
exists because  $ \phi \in  \mathcal{D}(  \mathrm{d} \Gamma (D) )  $  (see \eqref{F4} and above). 
\end{proof}              
\begin{lemma}\label{Invariance}
The  derivatives (recall \eqref{Vtilde})
\begin{align}\label{derivatives}
 \frac{\partial}{\partial \beta} \frac{1}{ H_{\overline P }^{(\beta)} - \lambda }|_{\beta = 0} &,  \hspace{1cm}   \frac{d}{d\beta}  \chi(H_{\overline P}^{(\beta)})|_{\beta = 0} ,
     \hspace{1cm}         \frac{\partial}{ \partial \beta}   (R^{\pm\eta}(z_{\pm\epsilon}))^{(\beta)} |_{\beta = 0},   
\\ \notag    \frac{\partial}{ \partial \beta} 
(   (\tilde V_{\eta})^*   R^{\mp\eta}(z_{\mp\epsilon}) )^{(\beta)}|_{\beta = 0} &, \hspace{1cm}
  \frac{\partial}{ \partial \beta}   (  \tilde V_{\eta}  R^{\pm\eta}(z_{\pm\epsilon}))^{(\beta)}  |_{\beta = 0}
\end{align}
exist, and therefore, the operators above leave $\mathcal{D}(  \mathrm{d} \Gamma (D) )$  invariant (see Lemma \ref{MAIN}). 
\end{lemma}
\begin{proof}
 In this proof, we denote by a dot on the top of a symbol the derivative with respect to $\beta$
 at zero. If it is  necessary, we specify below  with respect to which norm is the derivative taken. For example, the (point-wise) derivative of  $ u_\beta f $  with respect to  $\beta $  at $zero$ is denoted by $  \dot{(u_\beta f)} $. In  case that the dependence on $\beta$ is written as a superscript, we  sometimes omit the symbol $\beta$. For example the (point-wise) derivative of $  \xi^{(\beta)}  $ at zero is denoted by $ \dot{\xi} $.     

A simple calculation shows that
\begin{align}
\label{mal3}
&\norm{\beta^{-1}(f^{(\beta)}-f)-    \dot{(u_\beta f)} }\leq C|\beta|, \quad
\left| \beta^{-1}\left( \omega^{(\beta)}(k) -\omega(k) \right) -     \dot{\omega}(k)\right| \leq  C |\beta| \omega(k).
\end{align}
This together with  Proposition~\ref{PTOPNDJO}    (see also \eqref{Hdilated} and similar calculations) implies that 
\begin{align}\label{papa}
\Big \| \Big ( \frac{1}{\beta}( H_{\overline P }( \omega, f )^{(\beta)} - H_{\overline P}  )  - H_{\overline P }(\dot{\omega}, \dot{(u_\beta f)}) \Big ) \frac{1}{H_f + 1 }\Big \|    \leq  C |\beta|. 
\end{align}
Then, the second resolvent identity and Proposition  \ref{PTOPNDJO} imply that, for every 
$\lambda \in \mathbb{C}$ with not vanishing imaginary part (here we proceed as in \eqref{luisa} below), 
\begin{align}\label{Derq}
\frac{\partial}{\partial \beta} \frac{1}{ H_{\overline P }^{(\beta)} - \lambda }|_{\beta = 0} =
- \frac{1}{ H_{\overline P} - \lambda }  H_{\overline P }(\dot{\omega}, \dot{(u_\beta f)})\frac{1}{ H_{\overline P } - \lambda },
\end{align} 
and therefore, we obtain that the derivative in the left term of the first line in \eqref{derivatives} exists. 
Similar proofs (and formulas) hold for $  H_f \frac{1}{ H^{(\beta)} - \lambda }  $ and $  \frac{1}{ H^{(\beta)} - \lambda }  H_f  $. Eq.\ \eqref{Derq} and the second resolvent equation (used as in \eqref{luisa} below) allows us to analyze the resolvents in the integrand in the Helffer-Sj\"ostrand formula (\eqref{HS}, with $n > 3$) and get (see also  Proposition  \ref{PTOPNDJO})
\begin{align}\label{Dchi}
\frac{d}{d\beta} (H_f + 1) \chi(H_{\overline P}^{(\beta)})|_{\beta = 0} = \pi^{-1} \int_\C \mathrm{d}x \mathrm{d}y \, \partial_{\overline z} \,  \tilde \chi (z) \frac{\partial}{\partial \beta} (H_f +1 ) \frac{1}{ H_{\overline P}^{(\beta)} - \lambda },
\end{align}
where  $z=x+iy$. This implies that the derivative in the middle term  of the first line in \eqref{derivatives} exists.   Similarly as in \eqref{papa}, we obtain that 
\begin{align}\label{papa1}
\frac{d}{d \beta} H( \xi^{(\beta)},  u_{\beta} D f )^{(\beta)}
  \frac{1}{H_f + 1 }|_{\beta = 0}  = H(\dot{\xi}, \dot{(u_\beta D f)}) \Big ) \frac{1}{H_f + 1 } .   
\end{align}
Eqs.\ \eqref{Dchi} and \eqref{papa1} imply that $ (M^2)^{(\beta)} $ is differentiable with respect to $\beta$ at $\beta = 0$ (see \eqref{Hdilated}).  This and \eqref{papa} imply that
$  (H^{\pm \eta }_{\overline{P} })^{(\beta)} \frac{1}{H_f+ 1} $ is differentiable with respect to $\beta$ at $\beta = 0$.  Now we calculate the derivative of $ (H_f + 1)(R^{\pm\eta}(z_{\pm\epsilon}))^{(\beta)}  $ at zero using the second resolvent equation:
\begin{align}\label{luisa}
\frac{1}{\beta} &  ( H_f + 1)  \Big (  (R^{\pm\eta}(z_{\pm\epsilon}))^{(\beta)}  - R^{\pm\eta}(z_{\pm\epsilon}) \Big )
\\ \notag  & +  (H_f + 1) R^{\pm\eta}(z_{\pm\epsilon}) \Big [ \frac{\partial }{\partial \beta} (  H^{\pm \eta }_{\overline{P} })^{(\beta)} \frac{1}{H_f + 1 }|_{\beta = 0} \Big ]   (H_f + 1)  R^{\pm\eta}(z_{\pm\epsilon})
\\ \notag 
 = & \Big \{   ( H_f + 1)    R^{\pm\eta}(z_{\pm\epsilon})  \Big \} \Big \{ \Big (  
 \frac{1}{\beta}( H^{\pm \eta }_{\overline{P} }  -  (H^{\pm \eta }_{\overline{P} })^{(\beta)} ) \frac{1}{H_f + 1}
   +    \Big [ \frac{\partial }{\partial \beta} (  H^{\pm \eta }_{\overline{P} })^{(\beta)} \frac{1}{H_f + 1 }|_{\beta = 0}     \Big ) \Big \} 
 \\  \notag & \hspace{10cm}   \cdot \Big \{  ( H_f + 1 ) R^{\pm\eta}(z_{\pm\epsilon}) \Big \}
\\ &  + \Big \{  ( H_f + 1)    \big (  (R^{\pm\eta}(z_{\pm\epsilon}))^{(\beta)} -  R^{\pm\eta}(z_{\pm\epsilon})      \big )\Big \}  \Big \{  \frac{1}{\beta} \big (   H^{\pm \eta }_{\overline{P} }  -  (H^{\pm \eta }_{\overline{P} })^{(\beta)}\big ) 
 \frac{1}{H_f + 1} \Big \}  \notag 
 \\ \notag & \hspace{10cm} \cdot  \Big \{  ( H_f + 1 )   R^{\pm\eta}(z_{\pm\epsilon}) \Big \} . 
\end{align}
 It follows from Proposition  \ref{PTOPNDJO} and \eqref{TEo5iV} that $  ( H_f + 1 )   R^{\pm\eta}(z_{\pm\epsilon}) $ is bounded. This and the fact that 
$  (H^{\pm \eta }_{\overline{P} })^{(\beta)} \frac{1}{H_f+ 1} $ is differentiable with respect to $\beta$ at $\beta = 0$ imply that the first term in the right hand side side of  \eqref{luisa} tends to zero as $\beta$ goes to zero. The same arguments and the fact that  
\begin{align}\label{lui2}
( H_f + 1)     (R^{\pm\eta}(z_{\pm\epsilon}))^{(\beta)} =\Big ( ( H_f + 1  ) \frac{}{   ( H_f + 1  )^{(\beta)} } \Big )
\Big (   ( H_f + 1)     (R^{\pm\eta}(z_{\pm\epsilon})  \Big )^{(\beta)}
\end{align}
is uniformly bounded for small $\beta$  (see Proposition  \ref{PTOPNDJO} and \eqref{TEo5iV}) 
 imply that the second term in the right hand side of \eqref{luisa} is bounded (uniformly  with respect to $\beta$). Since the second term in the left hand side of \eqref{luisa} is bounded (see arguments above), it follows that  
\begin{align}\label{lui3}
\lim_{\beta \to 0} ( H_f + 1)  \Big (  (R^{\pm\eta}(z_{\pm\epsilon}))^{(\beta)}  - R^{\pm\eta}(z_{\pm\epsilon}) \Big ) = 0 .
\end{align}
This in turn and the  arguments above imply that the second term in the right hand side of \eqref{luisa}  tends to zero as $\beta $ tends to zero. We conclude that the left hand side of \eqref{luisa} tends to zero as $\beta$ tends to zero and, therefore,  $  ( H_f + 1)    (R^{\pm\eta}(z_{\pm\epsilon}))^{(\beta)}   $
 is differentiable at zero.  This proves the existence of the derivative in the right term of the first line in  \eqref{derivatives}.  The proof that the derivative of   $         (\tilde V_{\eta})^{(\beta)}  \frac{1}{ H_f +1}  $, with respect to $\beta$, at zero exists follows exactly the same lines as the corresponding result for     $  (H^{\pm \eta }_{\overline{P} })^{(\beta)} \frac{1}{H_f+ 1} $,  and therefore, we omit it.  Then, using this and that $  ( H_f + 1)    (R^{\pm\eta}(z_{\pm\epsilon}))^{(\beta)}   $
 is differentiable at zero, we obtain that $  \Big ( \tilde V_{\eta}  \frac{1}{1 + H_f} \Big )  \Big ( (1+ H_f)    R^{\pm\eta}(z_{\pm\epsilon})^{(\beta)} \Big )  $ is differentiable at zero. This proves the existence of the derivative in the right term of the second line in \eqref{derivatives}. The proof for the left term is analogous.    
\end{proof}

\subsection{Commutator Estimates in Mourre Theory  }

\begin{lemma}
\label{lemma2}
Recall  that  we introduce   $\chi\in C^{\infty}_c(\R,[0,1])$  in  Definition \ref{def.chi}.
  The 
quadratic form $  [\chi(H_{\overline P}),   \mathrm{d}\Gamma(D)]  $, defined in the domain of 
$  \mathrm{d}\Gamma(D)  $,  extends to a bounded operator that we denote by $   [\chi(H_{\overline P}),   \mathrm{d}\Gamma(D)]^0  $. Additionally, 
$ (H_{\overline P}+i)[\chi(H_{\overline P}),   \mathrm{d}\Gamma(D)]^0  $ is bounded.   
\end{lemma} 
\begin{proof}
For $\psi , \phi \in \mathcal D( \mathrm{d}\Gamma(D))\cap \mathcal D(H_{\overline P})$ and $z\in \C\setminus \R$, it follows from 
  Lemma  \ref{Invariance}    that
\begin{align}
\label{trick'}
\left\langle\phi, [(H_{\overline P}-z)^{-1},  \mathrm{d}\Gamma(D)]\psi \right\rangle 
&=\left\langle  \mathrm{d}\Gamma(D)    (H_{\overline P}-   \overline z   )^{-1}       \phi,    (H_{\overline P}-z)    (H_{\overline P}-z)^{-1}  \psi \right\rangle 
\notag \\
& - \left\langle   (H_{\overline P}-     \overline z    )      (H_{\overline P}-      \overline z    )^{-1}       \phi,    \mathrm{d}\Gamma(D) (H_{\overline P}-z)^{-1}  \psi \right\rangle
\notag \\
&= -   \left\langle   \phi, (H_{\overline P}-\overline z)^{-1}   [H_{\overline P},   \mathrm{d}\Gamma(D)]^0   (H_{\overline P}-z)^{-1} \psi \right\rangle 
.
\end{align}
Note that
\begin{align}
\label{trick1}
\norm{(H_{\overline P}+i) (H_{\overline P}-z)^{-1} }\leq 1+ \norm{(z+i) (H_{\overline P}-z)^{-1}} \leq C \left(1+ |\Re z|   |\Im z|^{-1} \right) .
\end{align}
Then, we observe from  Remark \ref{DEfiCommut}  that  
\begin{align}
\left|  \left\langle\phi, [(H_{\overline P}-z)^{-1},  \mathrm{d}\Gamma(D)]\psi \right\rangle  \right| \leq   C  \norm{\phi} \norm{\psi}  |\Im z|^{-1} \left(1+ |\Re z|   |\Im z|^{-1} \right)  ,
\end{align}
and consequently, $[(H_{\overline P}-z)^{-1},  \mathrm{d}\Gamma(D)]$ uniquely extends to a bounded operator on $\mathcal H$ which we denote by $[(H_{\overline P}-z)^{-1},  \mathrm{d}\Gamma(D)]^0$    and  
\begin{align}
\label{trick}
[(H_{\overline P}-z)^{-1},  \mathrm{d}\Gamma(D)]^0 =   -   (H_{\overline P}-z)^{-1}   [H_{\overline P},   \mathrm{d}\Gamma(D)]^0   (H_{\overline P}-z)^{-1} . 
\end{align}
 This together with  Remark \ref{DEfiCommut} ,  \eqref{trick1} and the Helffer-Sj\"ostrand formula    (see \eqref{HS})  yields 
 \begin{align}
 \label{hsf1}
&\norm{(H_{\overline P}+i)[\chi(H_{\overline P}),   \mathrm{d}\Gamma(D)]^0} 
\leq  \pi^{-1} \int_\C \mathrm{d}x \mathrm{d}y \, |\partial_{\overline z} \tilde \chi (z)| \norm{(H_{\overline P}+i) [(H_{\overline P}-z)^{-1},  \mathrm{d}\Gamma(D)]^0}
\notag \\
&     \leq    \pi^{-1} \int_\C \mathrm{d}x \mathrm{d}y \, |\partial_{\overline z} \tilde \chi (z)| 
    \norm{(H_{\overline P}+i) (H_{\overline P}-z)^{-1} }^2     \norm{  [H_{\overline P},   \mathrm{d}\Gamma(D)]^0   (H_{\overline P} -  i    )^{-1}}
\notag \\
&\leq  C \int_\C \mathrm{d}x \mathrm{d}y \, |\partial_{\overline z} \tilde \chi (z)|   \left( 1+|x| |y|^{-1} \right)^2
,
\end{align}
where we    take   $z=x+iy$ for $x,y\in \R$ and $\tilde \chi$ is the almost analytic extension of $\chi$ (see Definition \ref{def:almostana}). In the definition of $\tilde \chi$ we choose $n\geq 2$,  and therefore,    $ |\partial_{\overline z} \tilde \chi (z)| \leq C |\Im z|^{2} $.
Since   $\chi $ is compactly supported, then $\tilde \chi$ is also compactly supported. It follows that    
\begin{align}
\norm{(H_{\overline P}+i)[\chi(H_{\overline P}),   \mathrm{d}\Gamma(D)]^0}  
&\leq C       \int_{\supp( \tilde \chi)}    dxdy    |y|^{2} \left( 1+|x| |y|^{-1} \right)^2 \leq C.   
\end{align}
This completes the proof.
\end{proof}
\begin{lemma}
\label{commM}
 Recall that  we introduce   $\chi\in C^{\infty}_c(\R,[0,1])$  in  Definition \ref{def.chi} and $M^2$ in \eqref{MOURRE}.
The 
quadratic form $  [\mathrm{d}\Gamma(D),M^2] $, defined in the domain of 
$  \mathrm{d}\Gamma(D)  $,  extends to a bounded operator that we denote by $   [ \mathrm{d}\Gamma(D),M^2]^0   $.   Similarly, 
the 
quadratic form $[i \mathrm{d}\Gamma(D) ,\tilde V_{\eta} ]$  extends to a
 $H_{\overline P}$-bounded operator that we denote by $[i \mathrm{d}\Gamma(D) ,\tilde V_{\eta} ]^0$. 
\end{lemma}
\begin{proof} 
For $\phi, \psi \in \mathcal D( \mathrm{d}\Gamma(D))\cap \mathcal D(H_{\overline P})$, we 
 observe from    Lemma \ref{Invariance}   and the $H_{\overline P}$-boundedness of  $[H_{\overline P}, i   \mathrm{d}\Gamma(D)]^0$ that
 \begin{align}
\label{m1}
 & \left\langle  \mathrm{d}\Gamma(D)\phi,  M^2 \psi \right\rangle -\left\langle M^2 \phi,    \mathrm{d}\Gamma(D)\psi \right\rangle 
  \\\notag
&=   \left\langle [\chi(H_{\overline P}),  \mathrm{d}\Gamma(D)]   \phi,      [H_{\overline P}, i   \mathrm{d}\Gamma(D)]^0\chi(H_{\overline P})   \psi  \right\rangle 
+\left\langle  \mathrm{d}\Gamma(D) \chi(H_{\overline P})     \phi,      [H_{\overline P}, i   \mathrm{d}\Gamma(D)]^0\chi(H_{\overline P})   \psi  \right\rangle 
  \\\notag
& -\left\langle [H_{\overline P}, i   \mathrm{d}\Gamma(D)]^0\chi(H_{\overline P})   \phi,    [\chi(H_{\overline P}),  \mathrm{d}\Gamma(D)]     \psi  \right\rangle 
-\left\langle [H_{\overline P}, i   \mathrm{d}\Gamma(D)]^0\chi(H_{\overline P})   \phi,      \mathrm{d}\Gamma(D)   \chi(H_{\overline P})  \psi  \right\rangle 
.
\end{align}
It follows from Lemma \ref{lemma2}   and Remark  \ref{DEfiCommut}   that 
 \begin{align}
 \label{m3}
 \left|  \left\langle[ \mathrm{d}\Gamma(D),\chi(H_{\overline P})]   \phi,      [H_{\overline P}, i   \mathrm{d}\Gamma(D)]^0\chi(H_{\overline P})   \psi  \right\rangle \right| \leq C \norm{\phi} \norm{\psi}
 \end{align}
 and 
 \begin{align}
 \label{m2}
 \left|  \left\langle   [H_{\overline P},  i   \mathrm{d}\Gamma(D)]^0  \chi(H_{\overline P})  \phi,      [ \mathrm{d}\Gamma(D),\chi(H_{\overline P})]  \psi \right\rangle \right|  \leq C \norm{\phi} \norm{\psi}.
 \end{align}
Moreover, for $\varphi, \vartheta \in \mathcal F_\text{fin}[\mathfrak{h}_0]$, we   obtain   from Lemma \ref{lemma:funcprop} (iv) and (v) that
\begin{align}
&\left\langle  \mathrm{d}\Gamma(D) \varphi,      [H_{\overline P}, i   \mathrm{d}\Gamma(D)]^0\vartheta \right\rangle -\left\langle [H_{\overline P}, i   \mathrm{d}\Gamma(D)]^0 \varphi,      \mathrm{d}\Gamma(D)   \vartheta  \right\rangle 
\notag \\
&=\left\langle   \varphi,   [ \mathrm{d}\Gamma(D) , \left(  \mathrm{d}\Gamma_{\overline{P}}(\xi) +g\sigma_1  \Phi_{\overline{P}}(Df)   \right)]\vartheta \right\rangle 
=\left\langle   \varphi,   \left(  \mathrm{d}\Gamma_{\overline{P}}(\tilde \xi)    - i  g\sigma_1  \Phi_{\overline{P}}(D^2 f)   \right) \vartheta \right\rangle ,
\end{align}
where $\tilde \xi =[  D,\xi  ]$.   Direct calculations show that  $|\tilde \xi | \leq C \omega$  and $D^2f  \in \mathfrak{h}$.    
  Proposition \ref{PTOPNDJO} implies that   $ \left(  \mathrm{d}\Gamma_{\overline{P}}(\tilde \xi)     - i g\sigma_1  \Phi_{\overline{P}}(D^2 f)   \right)$ is relatively bounded with respect to $H_{\overline P}$ 
and, hence, $[ \mathrm{d}\Gamma(D) ,      [H_{\overline P}, i   \mathrm{d}\Gamma(D)]^0]$ extends to a $H_{\overline P}$-bounded operator  which we denote by $[ \mathrm{d}\Gamma(D) ,      [H_{\overline P}, i   \mathrm{d}\Gamma(D)]^0]^0= \mathrm{d}\Gamma_{\overline{P}}(\tilde \xi)       - i  g\sigma_1  \Phi_{\overline{P}}(D^2 f) $.
Employing    Lemma  \ref{Invariance},  we find a constant $C>0$ such that
\begin{align}
&\left| \left\langle  \mathrm{d}\Gamma(D) \chi(H_{\overline P})     \phi,      [H_{\overline P}, i   \mathrm{d}\Gamma(D)]^0\chi(H_{\overline P})   \psi  \right\rangle 
-\left\langle [H_{\overline P}, i   \mathrm{d}\Gamma(D)]^0\chi(H_{\overline P})   \phi,      \mathrm{d}\Gamma(D)   \chi(H_{\overline P})  \psi  \right\rangle \right|
\notag \\
&=\left| \left\langle\chi(H_{\overline P})     \phi,    \left[  \mathrm{d}\Gamma(D)  , [H_{\overline P}, i   \mathrm{d}\Gamma(D)]^0\right]^0 \chi(H_{\overline P})   \psi  \right\rangle  \right| \leq C\norm{\phi}\norm{\psi} .
\end{align}
This together with \eqref{m1}, \eqref{m2} and \eqref{m3} implies  that there is  a constant $C>0$ such that
\begin{align}
& \left| \left\langle  \mathrm{d}\Gamma(D)\phi,  M^2 \psi \right\rangle -\left\langle M^2 \phi,    \mathrm{d}\Gamma(D)\psi \right\rangle  \right| 
\leq C\norm{\phi} \norm{\psi},
\end{align}
and, thereby, we complete the proof, since $\mathcal D( \mathrm{d}\Gamma(D))\cap \mathcal D(H_{\overline P})$ is dense in $\mathcal H$. 
 The statement concerning $[i \mathrm{d}\Gamma(D) ,\tilde V_{\eta} ]$ is proved following the same lines above.  
\end{proof} 
\end{appendix}

\section*{Acknowledgement}
D.\ -A.\ Deckert and F.\ H\"anle would like to thank the IIMAS at UNAM and M.\
Ballesteros and J.\ Faupin  the Mathematisches Institut at LMU Munich   for their hospitality. This
project was partially funded by the DFG Grant DE 1474/3-1,  the grants PAPIIT-DGAPA
UNAM  IN108818, SEP-CONACYT 254062, and the junior research group ``Interaction
between Light and Matter'' of the Elite Network Bavaria. M.\  B.\  is a
Fellow of the Sistema Nacional de Investigadores (SNI). F.\ H.\ gratefully acknowledges financial support by the ``Studienstiftung des deutschen Volkes''.
Moreover, the authors express their gratitude for the fruitful discussions with
V.\ Bach,  J.\ S.\ M\o ller, A.\ Pizzo, W.\ De Roeck, R. Weder and P. Barberis.

\bibliographystyle{amsplain}
\bibliography{ref}
\end{document}